\documentclass[11pt,A4]{amsart}

\usepackage{graphicx,verbatim}
\usepackage{amssymb}
\usepackage{epstopdf}
\usepackage[usenames]{color}

\usepackage[normalem]{ulem}

\usepackage{xifthen}
\newboolean{ESA}
\setboolean{ESA}{true}
\newcommand{\Questions[1]}{\ifthenelse{\boolean{ESA}}{\textcolor{magenta}{#1}}{#1}}

\newcommand\prob{\mathbb{P}}

\theoremstyle{plain}
\newtheorem{thm}{Theorem}[section]
\newtheorem{lemma}[thm]{Lemma}
\newtheorem{prop}[thm]{Proposition}
\newtheorem{cor}[thm]{Corollary}

\theoremstyle{definition}

\newtheorem{ex}[thm]{Example}

\theoremstyle{remark}

\newcommand{\PP}{\mathbb P}

\newcommand{\MRCA}{\operatorname{MRCA}}
\newcommand{\desc}{\operatorname{desc}_\mathcal X}

\newcommand{\tc}{\text{:}} 

\allowdisplaybreaks

\title[Split probabilities]{Split probabilities and species tree inference under the multispecies coalescent model}

\date{\today}                    

\author{E. S. Allman}
\address{Department of Mathematics and Statistics, University of Alaska Fairbanks, PO Box 756660, Fairbanks, AK 99775 USA}
\email{e.allman@alaska.edu}

\author{J. H. Degnan}
\address{Department
of Mathematics and Statistics, The University of New Mexico, Albuquerque, NM  87131}
\email{jamdeg@unm.edu }

\author{J. A. Rhodes}
\address{Department of Mathematics and Statistics, University of Alaska Fairbanks, PO Box 756660, Fairbanks, AK 99775 USA}
\email{j.rhodes@alaska.edu}

\begin{document}


\begin{abstract}
Using topological summaries of gene trees as a basis for species tree inference is a promising approach to obtain acceptable speed on genomic-scale datasets, and 
to avoid some undesirable modeling assumptions.  Here we study the probabilities of 
splits on gene trees under the multispecies coalescent model, and how their features 
might inform species tree inference. After investigating the behavior of split consensus 
methods, we investigate split invariants --- that is, polynomial relationships 
between split probabilities. These invariants are then used to show that, even though a split is an unrooted notion, split 
probabilities retain enough information to identify the rooted species tree topology for
trees of more than 5 taxa, with one possible 6-taxon exception.\end{abstract}

\maketitle

\section{Introduction}

As advances in technology have allowed for the collection of genomic
scale data across a collection of organisms, it has been frequently
observed that phylogenetic trees inferred from single genes for a
fixed taxon set often differ from one another. Improving inference of
species relationships requires addressing such gene tree discordance
in a principled way. While there are a number of biological processes
that might cause this discord, including hybridization or other forms
of horizontal gene transfer, \emph{incomplete lineage sorting} is an
especially common source of gene tree incongruence when times between
speciation events are short and/or population sizes are large. Incomplete lineage sorting is modeled by the
the \emph{multispecies coalescent model}, an extension of the standard coalescent model 
describing gene tree formation within a single population.

Many methods of species tree inference based on the multspecies
coalescent have been proposed. The Baysian approaches of the software
*BEAST \cite{heled2010} and Mr.Bayes/BEST \cite{Liu2007,MrBayes}
perform simultaneous gene tree and species tree inference under a
combined coalescent and sequence evolution model. The SVDquartets
method \cite{SVDquartets} bypasses inference of individual gene trees,
yet still gives a statistically consistent estimate of the species
tree from sequence data. Alternatively, gene trees inferred by
traditional phylogenetic methods can be used as input for a
subsequent inference of a species tree, in
methods such as Rooted triple Consensus \cite{RTcon2008}, STEM
\cite{kubatko2009}, STAR \cite{Liu2009},
$\text{NJ}_{st}$/U-STAR/ASTRID  \cite{Liu2011,adr2016,ASTRID}, MP-EST
\cite{Liu2010}, BUCKy \cite{Larget2010} and ASTRAL-II \cite{ASTRALII}. While theoretical
justification for these two-stage approaches generally ignores gene
tree inference error, they can be applied to much larger data sets
(more taxa and more genes) than the computationally intensive Bayesian
algorithms, and have exhibited strong performance in simulations. Such
scalability makes them highly attractive, and motivates further
exploration of their underpinnings.

A fundamental issue for any inference of a species tree is how to
relate the time scale used in the multispecies coalescent model on a
species tree to those in sequence evolution models used on gene
trees. The coalescent time scale can be measured in number of generations divided by
population size, while the sequence evolution model is generally in
number of substitutions per site. Assumptions such as a constant
mutation rate over the gene tree (implying all gene trees are
ultrametric) and a constant population size over the species tree are
sometimes made, despite their implausibility. While these assumptions
can be relaxed somewhat through more elaborate modeling, it is
difficult to test the robustness of inference when they are violated.

An alternative way of addressing this difficult time scale issue is to
simply discard all metric information inferred about gene
trees, and only use their topological features to infer a species
tree.  Although discarding such information is undesirable if one can
validly relate time scales, one can view it as a conservative approach
to avoid reliance on unjustified assumptions.  

Some methods go
further, and only consider summaries of the inferred topological gene
trees, such as displayed quartets (unrooted 4-taxon trees), rooted
triples (rooted 3-taxon trees), clades (all taxa descended from an
internal node of the rooted tree) or splits (bipartitions of taxa
induced by an edge in the tree). 
Among current methods, Rooted triple consensus, MP-EST, BUCKy, STAR, ASTRAL-II, and $\text{NJ}_{st}$/U-STAR/ASTRID  are all of this type. (While for Rooted triple Consensus, MP-EST, BUCKy and ASTRAL-II this is obvious from their formulations, for STAR and $\text{NJ}_{st}$ the connection to clades and splits was established in \cite{adr2013} and \cite{adr2016}.) 
These methods all use topological summaries of a sample of gene trees, rather than the full gene trees, in ways allowing for statistically consistent species tree inference when the gene trees are sampled from the multispecies coalescent model without error.

\medskip

In this work, we undertake a theoretical study of the probabilities of
splits on gene trees arising from the multispecies coalescent, with
the aim of better understanding how species tree inference may be
performed from gene tree split information. This parallels several
previous works, in which we have shown that rooted species tree
topologies are identifiable from unrooted gene tree topologies or from
clades displayed on gene trees, and unrooted species tree topologies
are identifiable from gene tree quartets.

A pleasant outcome of our study is that gene tree split probabilities 
generally retain enough information on the species tree that they determine both
its topology and its root, despite the fact that splits themselves are
an unrooted notion. Specifically, for all species trees on 5 or more taxa, with one 6-taxon exception, for generic edge lengths, 
the rooted species tree topology is identifiable.
Translating to an empirical setting, this means
that one should be able to develop a statistically consistent method
of inference of the rooted species tree from the frequencies of splits
on a collection of unrooted gene trees. This would be an extension of
the $\text{NJ}_{st}$/U-STAR/ASTRID method, which only infers an
unrooted species tree topology from the same information.  Such a
method would avoid any issues with erroneous rooting of gene trees by
inclusion of an outgroup, which has long been known as a source of
additional error (see \cite{Philippe2011} for a recent discussion),
and allow rooting even when no appropriate outgroup is available. The
method by which we show the species tree root is identifiable depends
on certain linear relationships between split probabilities, and this
simple form gives hope that it can be developed into a well-founded
statistical test for root location.

\medskip

After setting notation in Section \ref{sec:notation}, we begin our
study of split probabilities under the multispecies coalescent in
Section \ref{sec:basic} with some basic observations. These include an
analysis of the behavior of greedy split consensus from gene tree
splits, concluding that it is not a statistically consistent method of
species tree inference even on trees with as few as 5 taxa. Detailed
arguments appear in Appendix \ref{app:greedy}.

Since split probabilities are complicated expressions that are
difficult to compute for trees with more than 6 taxa, in Section
\ref{sec:invariants} we turn our attention to relationships between
such split probabilities --- that is, rather than focus on
\emph{explicit} formulas for them, we look for \emph{implicit}
formulas they must satisfy. Our methods thus are mathematically the
same as those used for studying pattern probabilities under sequence
evolution models through \emph{phylogenetic invariants}, so we adopt
the same terminology of referring to equalities as invariants.  Our
previous work \cite{adr2016} on relationships between the split
probabilities and the $\text{NJ}_{st}$/U-STAR/ASTRID inference method
quickly leads to a number of linear invariants and inequalities the
split probabilities must satisfy, tied to the quartets displayed on
the species tree. These depend only on the unrooted species tree, and
thus give no information on its root. However, building on results in
\cite{adr2011b}, we then find additional split invariants that depend
on the clades displayed on the species tree, which thus give some
information about the root location. While our theoretical work gives
only linear invariants, higher degree ones also exist. Unfortunately a
computational determination of them was successful only for 5-taxon
trees, and their structure remains mysterious, but we report them in
Appendix \ref{app:grobner}.

In Section \ref{sec:ident} we build on the results on
linear invariants from Section 4, to prove a main result: the
collection of split probabilities under the multispecies coalescent
model determines the species tree topology, including the root
location (with one exception).  This \emph{identifiability} result, 
Theorem \ref{thm:rootID}, holds generically, i.e., for
all edge lengths on the species tree not in a set of measure zero.  For
most trees testing whether the invariants found in the previous
section vanish is sufficient for locating the root; however, for
certain trees these tests leave several possibilities for the
branching pattern near the root of the tree.  Motivated by known
invariants, we formulate some linear inequalities that resolve these
ambiguities in all cases, except for a particular unrooted 6-taxon
tree shape. Establishing that these inequalities hold is accomplished
by a laborious technical argument, which is relegated to Appendix
\ref{app:proofs}.

\section{Notation}\label{sec:notation}

Let $\mathcal X$ be a finite set of taxa, whose elements are denoted
by lower case letters $a,b,c,\dots$ etc. For any specific gene, we
denote a single sample from each taxon by the corresponding upper case
letter $A,B,C,\dots$ etc. If $\mathcal A\subseteq \mathcal X$ is a
subset of taxa, the corresponding subset of genes is $\mathcal
A_g\subseteq \mathcal X_g$.

By a species tree $\sigma=(\psi,\lambda)$ on $\mathcal X$ we mean a
rooted topological phylogenetic tree $\psi$, with leaves bijectively
labelled by $\mathcal X$, together with an assignment of edge weights
$\lambda$ to its internal edges. These edge weights are specified in
coalescent units, so that the multispecies coalescent model on
$\sigma$ leads to a probability distribution on rooted
gene trees with leaves labelled by $\mathcal X_g$. (For more on the multispecies coalescent model as we use it, see \cite{adr2011a}.)  Since we limit ourselves to
the situation where one individual is sampled per taxon, no coalescent
events can occur in pendant edges of a species tree, so the lengths of
those edges are inconsequential and
omitted from our notation. (If more than one individual is sampled per
taxon, one can create an ``extended species tree'' as in
\cite{adr2011b} by grafting several pendant edges of unspecified
length to the leaf labeled by that taxon, and assigning a length to
the formerly pendant edge, to again be in the framework set here.)

The gene trees sampled from the coalescent are rooted metric binary
phylogenetic trees on $\mathcal X_g$, though by marginalization over
edge lengths and root locations, this probability distribution also
leads to one on unrooted topological binary gene trees. Unrooted
topological gene trees will be denoted by $T$, and the probability of
an unrooted topological gene tree under the multispecies coalescent on
$\sigma$ is denoted by $\PP_\sigma(T)$, 
or simply $\PP(T)$ when $\sigma$ is clear from context.

\smallskip

A \emph{split} of a set of taxa $\mathcal X$ is a bipartition
$\mathcal A\sqcup \mathcal B$ into nonempty subsets, denoted $\mathcal
A|\mathcal B=\mathcal B|\mathcal A=Sp(\mathcal A)=Sp(\mathcal B)$.  If
$\sigma$ is a species tree on $\mathcal X$ then by a split \emph{on} $\sigma$
we mean a split of $\mathcal X$ formed by deleting a single edge of
$\psi$ and grouping taxa according to the connected components of the
resulting graph.  We similarly refer to splits of $\mathcal X_g$, and
splits of $\mathcal X_g$ on specific gene trees $T$. For small sets of
taxa, it will often be convenient to use juxtaposition of elements to
represent sets, rather than standard set notation. Thus $ac=\{a,c\}$
and $Sp(ac)=Sp(\{a,c\})$.

A \emph{trivial} split is one with one of the partition blocks a
singleton set. Trivial splits for taxa $\mathcal X$ appear on every
phylogenetic tree on $\mathcal X$.
For $\mathcal A\subset \mathcal X$ we will denote the
complementary set of $\mathcal A$ by $\overline {\mathcal A}=\mathcal
X\smallsetminus \mathcal A$, so that $\mathcal A|\overline {\mathcal A}$ is a split when
$\emptyset\neq \mathcal A\subsetneq \mathcal X$.

For a species tree $\sigma$ on $\mathcal X$, by the probability of a
split $\mathcal A|\mathcal B$ of $\mathcal X$ under the multispecies
coalescent we mean
\begin{equation}\PP_\sigma(\mathcal A|\mathcal B)=\sum_{T} \PP_\sigma(T) \delta_{\mathcal A|\mathcal B}(T) \label{eq:defsplitprob}
\end{equation}
where $\delta_{\mathcal A|\mathcal B}(T)$ is 1 if $\mathcal
A_g|\mathcal B_g$ is a split on $T$, and 0 otherwise, and the sum runs
over all binary unrooted topological phylogenetic trees on $\mathcal
X_g$.  Thus the probability of a split is the probability that an
observation of a gene tree displays the corresponding split.  Note
that trivial splits have probability 1 for every species tree, since
they are on every binary gene tree.

We will also need to refer to clades and quartets of taxa. A \emph{clade} is simply a subset 
$\mathcal A\subseteq \mathcal X$. A clade $\mathcal A$ is \emph{on} the species tree $\sigma$ 
if it equals the set of all leaf-descendants of some node in the tree. A \emph{quartet} is a 4-element 
subset of $\mathcal X$ partitioned into 2-element sets, denoted as
$ab|cd$, with $a,b,c,d\in \mathcal X$. A quartet $ab|cd$ is \emph{on} $\sigma$ if the unrooted 
tree with leaves labeled $a,b,c,d$ induced from $\sigma$ has an edge separating $a,b$ from $c,d$.
 
\section{Basic observations}\label{sec:basic}

While in principle it is straightforward to compute the probabilities
of gene tree splits for a fixed species tree, in practice the work
required can be formidable. For an $n$-taxon species tree, using the
definition in equation \eqref{eq:defsplitprob}, one first must compute
probabilities of each of the $(2n-3)!! = 1 \cdot 3 \cdots (2n-3)$
unrooted topological gene trees. This can be accomplished by work of
\cite{DegnanSalter2005} or \cite{Wu2012} in finding probabilities of
all rooted topological gene trees, and then marginalizing over the
root locations. For a given split one must still sum over all unrooted
gene trees displaying that split. If the split has blocks of size $k$
and $n-k$, then there are $(2k-3)!! \, (2n-2k-3)!!$ such unrooted
trees.

Using this approach, we computed split probabilities for all species
trees on 6 or fewer taxa for use in computations discussed in later
sections, but went no further. Indeed, this approach does not seem to 
be tractable except for small species trees. On the other hand,
since the U-STAR inference methods implemented in ASTRID are based on
split frequencies \cite{ASTRID,adr2016} and perform well on large data
sets, theoretical study of these probabilities is still strongly
warranted.

As a first step, analogous to Proposition 1 of \cite{adr2011b} for
clade probabilities, we have the following.

\begin{lemma}\label{lem:trivialinv}
If $\vert \mathcal X \vert = n$, then the sum of the non-trivial split probabilities
is $n-3$.
\end{lemma}

\begin{proof}  
First considering all splits of $\mathcal X$, including trivial ones,
\begin{align*}
\sum_{\mathcal A|\mathcal B} \PP_\sigma({\mathcal A|\mathcal B} )=
&\sum_{\mathcal A|\mathcal B} \sum_{T} \PP_\sigma(T ) \delta_{\mathcal A|\mathcal B}(T)\\
&=
\sum_T   \PP_\sigma(T )  \sum_{\mathcal A | \mathcal B}  \delta_{\mathcal A|\mathcal B}(T) \\
&=
\sum_T   \PP_\sigma(T ) (2n-3) =2n-3.
\end{align*}
Since the $n$ trivial splits of $\mathcal X$ all have probability 1, removing them from the sum gives the claim.
\end{proof}

Another analog of a result for clade probabilities, Theorem 3 of
\cite{adr2011b}, is the content of the next Proposition.  T.~Warnow first asked if this
might hold, and C.~An\'e independently provided a proof \cite{AnePC}.

\begin{prop}
  Let $\sigma$ be a binary species tree on $\mathcal X$, with internal
  edge lengths $\lambda_i> \epsilon\ge 0$, and $\mathcal A|\mathcal B$
  a split of $\mathcal X$.  Then under the multispecies coalescent
  model if
$$\PP_\sigma(\mathcal A|\mathcal B)\ge(1/3)\exp(-\epsilon)$$ then
$\mathcal A|\mathcal B$ is a split on $\sigma$. 

Furthermore, if $(1/3)\exp(-\epsilon)$ is replaced with any smaller
number, this statement is no longer true: For any $\alpha<
(1/3)\exp(-\epsilon)$, there exists a species tree $\sigma$
with branch lengths $\lambda_i > \epsilon \ge 0$ and a
split $\mathcal A|\mathcal B$ of $\mathcal X$ not displayed on
$\sigma$ with $\PP_\sigma(\mathcal A|\mathcal B)>\alpha$.

\end{prop}

\begin{proof} The first statement holds for trivial 
splits, since they have probability 1 and are displayed on every binary $\sigma$. 

Now
consider a non-trivial split $\mathcal A|\mathcal B$ not displayed on $\sigma$.
Then there exist $a_1, a_2\in \mathcal A$, $b_1,b_2\in \mathcal B$ so the quartet $a_1a_2|b_1b_2$ is not displayed on $\sigma$. Thus, by 
\cite[Section 4.1]{adr2011a} 
the probability that an unrooted gene tree displays the quartet $A_1A_2|B_1B_2$ is
$$\PP_\sigma(A_1A_2|B_1B_2)=(1/3)\exp(-\ell)<(1/3)\exp(-\epsilon)$$
where $\ell>\epsilon$ is the sum of the lengths of all branches in
$\sigma$ that form the central edge in the induced quartet tree on
$a_1,a_2,b_1,b_2$.  But since displaying the split $\mathcal
A_g|\mathcal B_g$ is a subevent of displaying $A_1A_2|B_1B_2$, this
implies that if $\mathcal A|\mathcal B$ is not displayed on $\sigma$
then
$$\PP_\sigma(\mathcal A|\mathcal B)<(1/3)\exp(-\epsilon),$$ establishing the first claim.

For the second claim, we construct an example.  For any non-trivial
split $\mathcal A|\mathcal B$, pick $a\in \mathcal A$, $b\in \mathcal
B$, and let $\mathcal A'=\mathcal A\smallsetminus \{a\}$, $\mathcal
B'=\mathcal B\smallsetminus \{b\}$. Pick any binary rooted tree
$\sigma_1$ on $\mathcal A'$, and any binary rooted tree $\sigma_2$ on
$\mathcal B'$, with internal branch lengths greater than $\epsilon$,
and consider the tree $$\sigma=( (
(a,b):\lambda_1,\sigma_1:\lambda_2):\lambda_3,\sigma_2:\lambda_4).$$

Note $\mathcal A|\mathcal B$ is not a split on this tree, yet if
$\lambda_2,\lambda _4$ are sufficiently large,  so that
$\mathcal A'_g$ and $\mathcal B'_g$ are almost certainly clades on a
gene tree, then the probability of a gene tree displaying the split
$\mathcal A_g | \mathcal B_g$ can be made arbitrarily close to
$(1/3)\exp(-\lambda_1)$. If $\alpha<(1/3)\exp(-\epsilon)$, there is a
choice of $\lambda_1>\epsilon$ so that
$\alpha<(1/3)\exp(-\lambda_1)$. Thus we can ensure
$\PP_\sigma(\mathcal A|\mathcal B) >\alpha.$
\end{proof}

\begin{cor} Suppose $\sigma$ is a binary species tree on $\mathcal X$,
  with positive edge lengths, and $\mathcal A|\mathcal B$ a split of
  $\mathcal X$.  Then under the multispecies coalescent model if
$$\PP_\sigma(\mathcal A|\mathcal B)\ge1/3$$ then
$\mathcal A|\mathcal B$ is a split on $\sigma$. 
\end{cor}
\begin{proof}
Set $\epsilon=0$ in the preceding theorem.
\end{proof}

This proposition has implications for a greedy split consensus
approach to inferring splits in a species tree. Recall that in this
method, one first orders splits observed in a gene tree sample by
decreasing frequency, arbitrarily (or randomly) breaking ties if
necessary. Proceeding in order down the list, splits are accepted if
they are compatible with all previously accepted ones.  For a large sample of 
gene trees from the multispecies coalescent,
a fully-resolved unrooted tree is likely to be returned, since all splits have positive probability.
The above corollary
implies that if one only allows the acceptance of splits of frequency
greater than 1/3, then this method will not be misleading; as the size
of the gene tree sample grows, the probability of accepting only
splits on the species tree goes to 1. While a tree displaying the
accepted splits may not be fully resolved, one can have confidence in
the splits that are displayed.

\smallskip

To show that accepting splits below a frequency 1/3 cutoff in greedy
split consensus would not lead to consistent species tree inference,
we investigate 5-taxon trees in more detail. Up to permutation of
taxon labels, there are three species trees to consider:
\begin{center}
\begin{tabular}{ll}
    balanced tree & $\sigma_{bal}=(((ab)\tc x, c)\tc y,(de)\tc z)$\\    
    pseudocaterpillar tree &$\sigma_{ps}= (((ab)\tc x, (de)\tc y)\tc z,c)$ \\
           caterpillar tree & $\sigma_{cat}=((((ab)\tc x,c)\tc y,d)\tc z,e)$
\end{tabular}
\end{center}
Although we use the same variables $x,y,z$ to denote the three
internal edge lengths in each tree, note that these have no
relationship across the species trees.  All split probabilities can be
expressed as polynomials in the transformed edge lengths
$$X = \exp(-x),\  Y = \exp(-y),\ Z = \exp(-z).$$
Note that with this transformation, values of $X$ close to
$1$ correspond to small branch lengths $x$, and values of $X$ close to
$0$ correspond to large branch lengths $x$.

The following two propositions are proved in Appendix \ref{app:greedy}.

\begin{prop}  \label{prop:greedynoncat} For the $5$-taxon balanced and pseudocaterpillar species trees, 
$\sigma = \sigma_{bal}$, 
$\sigma_{ps}$,  with positive branch lengths,
$$\prob_\sigma (Sp(ab)),\ \prob_\sigma (Sp(de)) > \prob_\sigma (\mathcal S)$$ 
for each of the
eight other non-trivial splits $\mathcal S$, so the splits displayed on the species tree have the 
highest probability of appearing on gene trees. 

When restricted to these species trees, as the sample size goes to
infinity greedy split consensus infers the correct unrooted species
tree topology with probability approaching 1.
\end{prop}

\smallskip

\begin{prop} \label{prop:greedycat} For the 5-taxon caterpillar
  species tree $\sigma = \sigma_{cat}$ with positive branch lengths,
  $\prob_\sigma (Sp(ab)) > \prob_\sigma (\mathcal S)$ for all
  non-trivial splits $\mathcal S \neq Sp(de)$, and $\prob_\sigma
  (Sp(de)) >\prob_\sigma(Sp(ce))$.

  If $\prob_\sigma (Sp(de))>\prob_\sigma (Sp(cd))$ for such a species
  tree, as the sample size goes to infinity greedy split consensus
  infers the correct unrooted species tree topology with probability
  approaching 1.

  However, if $\prob_\sigma (Sp(de))<\prob_\sigma (Sp(cd))$, it infers
  the incorrect unrooted species tree topology $((a,b),e,(c,d))$ with
  probability approaching 1.  The parameter region in which this
  occurs is
\begin{equation}  \label{E:too-greedy}
18 + XY^3Z^6 + 2XY^3 - 3XY - 18 Y <0.
\end{equation}
\end{prop}

\begin{center}
\begin{figure}[h]
\includegraphics[width=10cm]{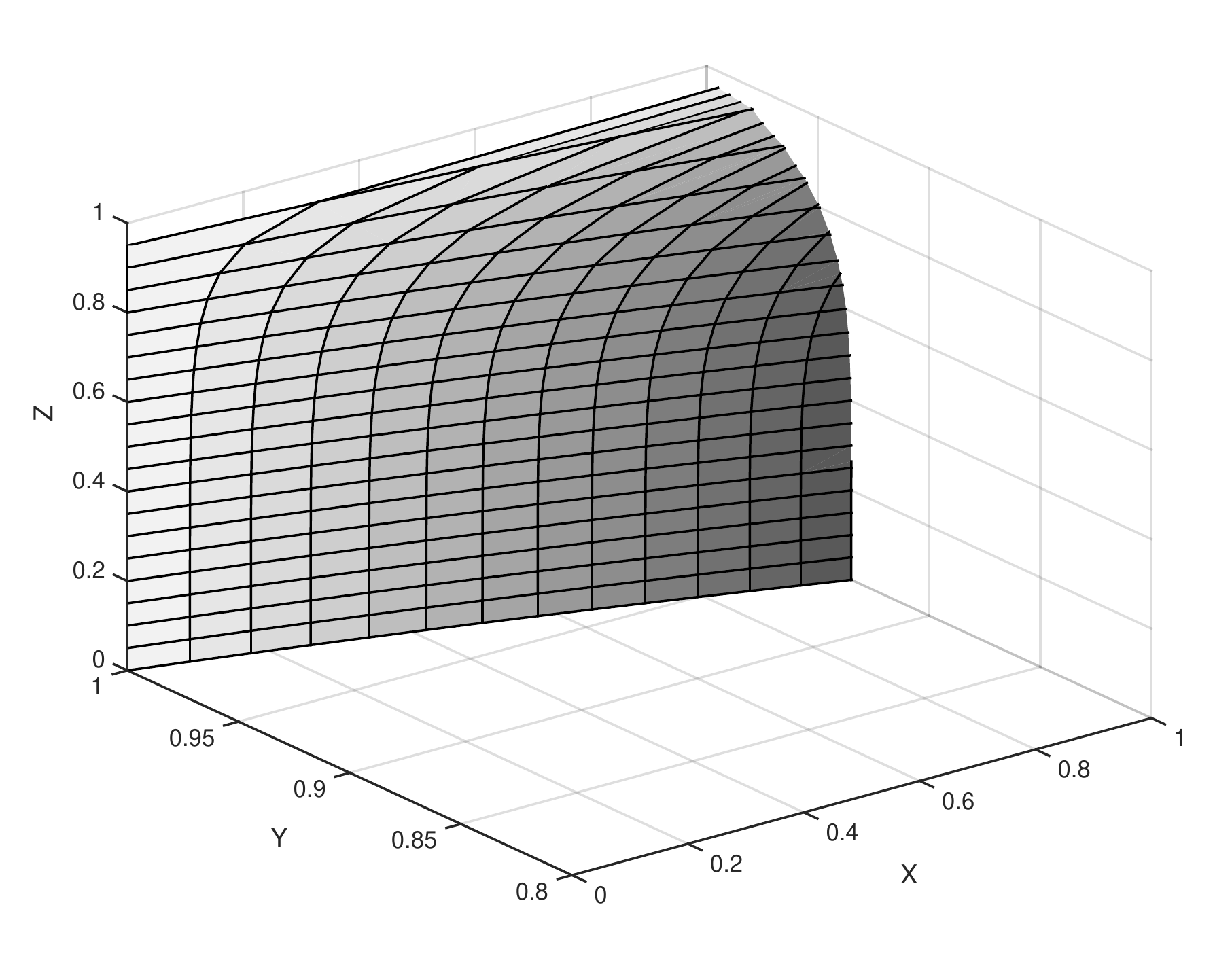}
\caption{The boundary of the too-greedy zone for the 5-taxon caterpillar tree $\sigma_{cat}$.  Greedy consensus
for split probabilities is inconsistent for all choices of parameters behind
the surface, and consistent in front of the surface.  For example, if $(X,Y,Z) \approx (.9, .96, .5)$, 
so species tree branch lengths are $(x,y,z) \approx  (0.1054,    0.0408,    0.6931)$,
then greedy
consensus with a large number of gene trees is expected to return the incorrect tree $((a,b),e,(c,d))$.} \label{fig:toogreedy}
\end{figure}
\end{center}

Figure \ref{fig:toogreedy} shows the surface dividing the regions of
parameter space on which greedy split consensus is misleading 
from that on which it is not.  We refer to the region
behind the surface, in which greedy consensus on splits is expected to
return the incorrect species tree, as the \emph{too-greedy zone}.

The analogous expression for the boundary of the five-taxon unrooted anomaly zone (the branch lengths for which the most likely unrooted gene tree does not match the unrooted caterpillar species tree) is (\cite{degnan2013anomalous}, equation (4))
\begin{equation}\label{E:uaz}
18+XY^3Z^6+2XY^3+9XY-12X-18Y<0.
\end{equation}
We note that if inequality \eqref{E:too-greedy} holds, then inequality \eqref{E:uaz} holds as well.  This means that for the five-taxon caterpillar, the too-greedy zone is a subset of the unrooted anomaly zone.  This relationship is also true for the rooted 4-taxon caterpillar case: the rooted too-greedy zone  is a subset of the rooted anomaly zone.  
In both the rooted and unrooted cases for four and five taxa, respectively, when greedy consensus is misleading, it returns the anomalous gene tree.
We leave it as an open question whether the too-greedy zones for larger trees are also subsets of the corresponding anomaly zones.

One can show that the minimum value of $Y$ on the boundary surface in
Figure \ref{fig:toogreedy} occurs when $X=1$, $Y \approx 0.93498735$,
and $Z=0$,  so
 $x = \infty$, $y \approx
0.06722228$, $z=0$.  For values of $y$
larger than this, regardless of the values of $x$ and $z$, branch
length parameters are outside the too-greedy zone, and greedy split
consensus is expected to return the correct species tree.

Moreover, if $Z = 0$ so that $z$ is an infinite branch length in
$\sigma_{cat}$, the too-greedy zone for splits coincides exactly with
the too-greedy zone for clade consensus on the $4$-taxon tree
$(((a,b),c),d)$ \cite{Degnan-etal-2009} .  This is as expected, since
placing the root of the 5-taxon caterpillar species tree ``at
infinity'' makes non-trivial splits for it correspond exactly to
clades in the 4-taxon caterpillar, by viewing $e$ as an outgroup
and noting that the lineage from $e$ must coalesce last.

\begin{figure}[h]
\begin{center}
\includegraphics[height=2.5in]{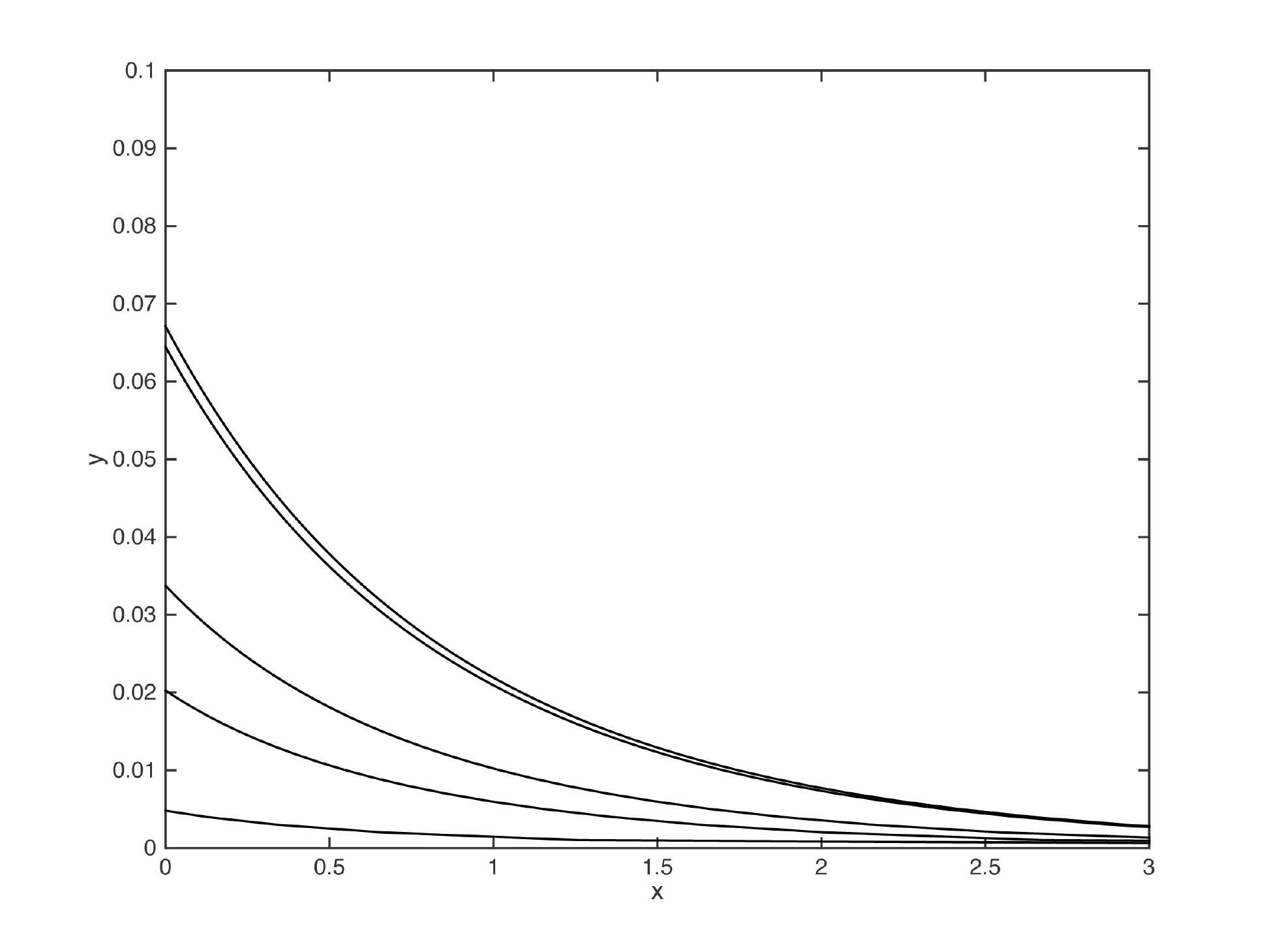}
\caption{For $z=0.01$, $.05$, $.1$, $.5$, $1$, curves giving, from
  bottom to top, the boundary of the too-greedy zone in the $xy$-plane
  for the 5-taxon caterpillar tree. Note $x,y,z$ are edge lengths in
  coalescent units. The too-greedy zone for a value of $z$ is the
  region below the curve, illustrating that the region grows with
  $z$. For $z>1$ the curve is visually indistinguishable from that for
  $z=1$.}\label{fig:toogreedy2}
\end{center}
\end{figure}

As was pointed out to us \cite{AnePC}, the shape of the
surface in Figure \ref{fig:toogreedy} has an interesting
consequence. As $Z$ increases from 0 to 1, the too-greedy zone in the
$XY$-plane becomes smaller.  Equivalently, in terms of branch lengths,
as $z$ increases from 0 to $\infty$, the too-greedy zone in the $xy$
plane becomes larger, as shown in Figure \ref{fig:toogreedy2}. Thus
smaller values of $z$ result in greedy split consensus performing well
for more choices of branch lengths $x$ and $y$.
 
If one views the taxon $e$ as an outgroup on $\sigma_{cat}$, this
means that an outgroup that is closely related to all other taxa
results in better performance of greedy split consensus than one that
is more distantly related.  While an extremely distantly related
outgroup ($z=\infty$) enables determination of the root of each gene
tree, so that knowing the splits on 5-taxon gene trees is equivalent
to knowing the clades on 4-taxon trees omitting the outgroup, this
actually reduces the ability of greedy split consensus to determine
the correct unrooted topology of the species tree. Indeed, the
too-greedy zone for clade consensus on a 4-taxon caterpillar as shown
in Figure 3 of \cite{Degnan-etal-2009} matches precisely that for
$z=\infty$ in Figure \ref{fig:toogreedy2}. 

These comments have analogs for the unrooted anomaly zones for five taxa.  In particular, the left side of inequality \eqref{E:uaz} is decreasing in $z$, so that increasing that branch length results in more values of $x$ and $y$ with $(x,y,z)$ in the anomaly zone.

\section{Linear invariants and inequalities for split probabilities}\label{sec:invariants}

A useful concept for understanding probabilistic models in
phylogenetics has been that of a \emph{phylogenetic invariant}. An
invariant of this sort is a multivariate polynomial that when
evaluated at probabilities arising from the model gives 0, regardless
of the particular parameter values associated to the model
instance. Equivalently, it is a polynomial relationship between the
probabilities that holds for all parameter values, and thus gives
information about the probabilities implicitly.

A \emph{split invariant} for a species tree topology is a polynomial
in the probabilities of splits under the multispecies coalescent model
that vanishes for all edge length assignments to the species tree.
More completely, a split invariant associated to an $n$-taxon species
tree topology $\psi$ is a multivariate polynomial in $2^{n-1}-1$
indeterminates (one for every split) which evaluates to zero at any
vector of split probabilities $\PP_\sigma(\mathcal A|\mathcal B)$
arising from $\sigma=(\psi, \lambda)$, regardless of the values of
$\lambda$.  Since probabilities of trivial splits are always 1 under
the coalescent, we can and will consider only invariants in variables
for the $2^{n-1}-n-1$ non-trivial splits.

The \emph{trivial split invariant}, which is valid for all $\psi$, is
\begin{equation}\label{eq:trivialinv}
n-3-\sum_{\text{non-trivial}\atop \text{splits $\mathcal A|\mathcal B$} }\PP_{\sigma} (\mathcal A|\mathcal B).
\end{equation} 
That this evaluates to 0 when the split probabilities arise from the
multispecies coalescent on some species tree is established by Lemma
\ref{lem:trivialinv}.

\subsection{Linear invariants and inequalities for unrooted species trees}

In this section we explore linear relationships, both invariants and
inequalities, between split probabilities that are tied to the U-STAR
algorithm \cite{adr2016}. Since U-STAR allows a distance method to be
used to infer a species tree, these have a rather direct
correspondence to equalities and inequalities defining tree metrics.

From equation (7) of \cite{adr2016}, we know the expected U-STAR
distance on a species tree $\sigma$ under the multispecies coalescent
can be expressed as
$$D(a,b)=\sum_{\text{splits } \mathcal A|\mathcal B \atop \text{separating } a,b} \PP_\sigma(\mathcal A|\mathcal B),$$
where a split is said to \emph{separate} two taxa if they lie in
different sets of the bipartition. Moreover, the expected U-STAR
distance is a tree metric on $\sigma$.  Using this in the 4-point
condition for tree metrics implies a collection of linear equalities
and inequalities in split probabilities that must hold for
\emph{unrooted} species trees. That is, these equalities and
inequalities hold regardless of the location of the root on the
species tree.

To state the result, we will say
that a split $\mathcal A|\mathcal B$ of $\mathcal X$ \emph{separates} two non-empty disjoint subsets $\mathcal Y_0,\mathcal Y_1\subset\mathcal X$ provided
$\mathcal Y_i\subseteq\mathcal A$ and $\mathcal Y_{1-i}\subseteq\mathcal B$ for one of $i=0,1$.

\begin{thm}\label{thm:4ptinv}
  Suppose $\sigma$ is an $n$-taxon species tree on $\mathcal X$, and
  $a,b,c,d\in\mathcal X$ are any four taxa for which $\sigma$ induces
  the quartet tree $ab|cd$. Then
$$\sum_{\text{splits } \mathcal A|\mathcal B\atop \text{separating } ac,bd}\PP_\sigma(\mathcal A|\mathcal B)=
\sum_{\text{splits } \mathcal A|\mathcal B\atop \text{separating } ad,bc}\PP_\sigma(\mathcal A|\mathcal B)\le
\sum_{\text{splits } \mathcal A|\mathcal B\atop \text{separating } ab,cd}\PP_\sigma(\mathcal A|\mathcal B).
$$
\end{thm}
\begin{proof} If $a,b,c,d$ are 4 taxa on a metric tree displaying the
  quartet $ab|cd$, then for the associated tree metric $d$ the 4-point
  condition \cite{SempleSteel} states that
$$d(a,c)+d(b,d)=d(a,d)+d(b,c)\ge d(a,b)+d(c,d).$$
The equality of the 4-point condition applied to the U-STAR expected distance gives
\begin{multline*}
\sum_{\text{splits } \mathcal A|\mathcal B \atop \text{separating } a,c} \PP_\sigma(\mathcal A|\mathcal B) +\sum_{\text{splits } \mathcal A|\mathcal B \atop \text{separating } b,d} \PP_\sigma(\mathcal A|\mathcal B)=\\
\sum_{\text{splits } \mathcal A|\mathcal B \atop \text{separating } a,d} \PP_\sigma(\mathcal A|\mathcal B)+\sum_{\text{splits } \mathcal A|\mathcal B \atop \text{separating } b,c} \PP_\sigma(\mathcal A|\mathcal B).
\end{multline*}
Any split for which all four of $a,b,c,d$ appear in the same split set
does not appear in this equation. If one of $a,b,c,d$ is separated
from the other three in a split, then that split probability occurs
exactly once on each side of the equation, and can be cancelled. If
two of $a,b,c,d$ are separated from the others, several cases must be
considered. First, if $ab$ is separated from $cd$ by a split, that
split probability occurs in all four sums, and so can be
cancelled. Second, if $ac$ is separated from $bd$ by a split, that
split probability occurs in both sums on the right, and not on the
left. Third, if $ad$ is separated from $bc$ by a split, that split
probability occurs in both sums on the left, and not on the
right. Thus after all canceling and division by 2 we obtain the
claimed equality.

The inequality of the 4-point condition is 
\begin{multline*}
  \sum_{\text{splits } \mathcal A|\mathcal B \atop \text{separating } a,d} \PP_\sigma(\mathcal A|\mathcal B)+\sum_{\text{splits } \mathcal A|\mathcal B \atop \text{separating } b,c} \PP_\sigma(\mathcal A|\mathcal B) \ge \\
  \sum_{\text{splits } \mathcal A|\mathcal B \atop \text{separating }
    a,b} \PP_\sigma(\mathcal A|\mathcal B) +\sum_{\text{splits }
    \mathcal A|\mathcal B \atop \text{separating } c,d}
  \PP_\sigma(\mathcal A|\mathcal B).
\end{multline*}
By reasoning similar to before, only probabilities of splits
separating two of $a,b,c,d$ from the others remain after cancellation.
Those separating $ac$ and $bd$ occur in all four sums, and hence can
be cancelled. Those separating $ab$ and $cd$ occur twice on the left
side, and those separating $ad$ and $bc$ occur twice on the
right. After all cancellation and division by 2 we have
$$\sum_{\text{splits } \mathcal A|\mathcal B \atop \text{separating } ab,cd } \PP_\sigma(\mathcal A|\mathcal B)\ge
\sum_{\text{splits } \mathcal A|\mathcal B \atop \text{separating }
  ad,bc }\PP_\sigma(\mathcal A|\mathcal B),$$ as claimed.
\end{proof}

Since the invariants of Theorem \ref{thm:4ptinv} depend only on
displayed quartets, they hold for all rooted versions of a fixed
unrooted species topology.

\begin{ex}\label{ex:5splits}  We apply the theorem to a species tree with unrooted 
topology $((a,b),c,(d,e)).$  With
\begin{align*}
s_1&=\PP(Sp(ab)),& s_2&=\PP(Sp(ac)),& s_3&=\PP(Sp(ad)),&s_4&=\PP(Sp(ae)),\\
s_5&=\PP(Sp(bc)),&s_6&=\PP(Sp(bd)),& s_7&=\PP(Sp(be)),& s_8&=\PP(Sp(cd)),\\
& & s_9&=\PP(Sp(ce)),&s_{10}&=\PP(Sp(de)),& &
\end{align*}
by Theorem \ref{thm:4ptinv} for each quartet on the species tree we
obtain an equality and inequality:
\begin{align*}
ab|cd:\ \ & s_3+s_5=s_2+s_6\le s_1+s_8,\\
ab|ce:\ \ & s_4+s_5=s_2+s_7\le s_1+s_9,\\
ab|de:\ \ & s_3+s_7=s_4+s_6\le s_1+s_{10},\\
ac|de:\ \ & s_3+s_9=s_4+s_8\le s_2+s_{10},\\
bc|de:\ \ & s_6+s_9=s_7+s_8\le s_5+s_{10}.
\end{align*}
The first three equalities here span a space of dimension 2, as do the
last three. The middle three are a basis for the span of them all.

One can compute all split probability invariants for the unrooted
5-taxon tree in {\tt Singular} \cite{Singular}, by computing and
intersecting the ideals of invariants for the 7 rooted versions of the
tree. Doing so shows that the invariants given here span the full
space of linear invariants for the unrooted tree. There is also a
quadratic invariant in a Gr\"obner basis for the ideal,
$$(s_2-s_5)(s_3-s_4)=0,$$
which reflects the fact that for each of the rooted trees either
$s_2=s_5$ or $s_3=s_4$.  As will be explained in the next section,
these equalities arise as cherry-swapping invariants, since each of
the 7 trees either has $ab$ or $de$ as a 2-clade. In addition there
are 14 higher degree invariants (not shown) in the basis for the
ideal, of total degree ranging from 3 to 8.
\end{ex}

\begin{ex} For each of the 2 unrooted shapes of binary 6-taxon trees
  one can similarly compute all linear split invariants.  For the
  unrooted tree shape with 2 cherries, exemplified by
  $(((a,b),c),d,(e,f))$, there is one additional linear split
  invariant, outside the span of those given by Theorem
  \ref{thm:4ptinv}:
$$s_{ace|bdf}-s_{acf|bde}+s_{ade|bcf}-s_{adf|bce}.$$
This also can be explained by the cherry-swapping invariants of the
next section, since any rooted version of this tree will have at least
one of $ab$ or $ef$ as a 2-clade.

For the unrooted shape with 3 cherries, exemplified by
$((a,b),(c,d),(e,f))$, in addition to the linear invariants of the
above theorem one finds
$$s_{ace|bdf}=s_{acf|bde}=s_{ade|bcf}=s_{adf|bce}.$$
All equalities can be explained by the fact that any rooted version of
the tree has at least 2 of the 2-clades $ab$, $cd$, and $ef$, and
determining the cherry swapping invariants these clades imply.
\end{ex}

\subsection{Linear invariants for rooted species trees}\label{sec:inv}

Next we investigate linear split invariants that depend on the rooted
species tree. More specifically, we construct a family of such
invariants associated to each non-trivial clade on the species tree.
The existence of these \emph{clade-induced split invariants}, as given
in Theorem \ref{thm:cladesplitinv} below, will form the basis of
arguments in Section \ref{sec:ident} that the root of the species tree
can be identified from split probabilities.

\begin{thm}\label{thm:cladesplitinv}
 Let $\mathcal A\subset \mathcal X$ be a subset of taxa with $|\mathcal A|,|\overline {\mathcal A}|\ge 2$.
Choose $\emptyset\ne \mathcal C\subsetneq   \overline {\mathcal A}$, and   distinct
$a,b\in \mathcal A$. Let $\mathcal A'=\mathcal
  A\smallsetminus\{a,b\}$. 
  
  Then if $\mathcal A$ is a clade on
  $\sigma$,
\begin{equation}
    \sum_{\mathcal S\subseteq \mathcal A'} \PP_\sigma(Sp(\mathcal S\cup\{a\}\cup \mathcal C))
  - \sum_{\mathcal S\subseteq \mathcal A'}\PP_\sigma(Sp(\mathcal S\cup\{b\}\cup \mathcal C))=0.
\label{eq:cladesplitinv}
\end{equation}
\end{thm}

\smallskip

 We note that this theorem applies to any species tree, including
 non-binary ones.   Moreover, since a non-binary species tree $\sigma$ can
  be thought of as any of its binary resolutions with length 0
  assigned to any introduced edges, the clade probabilities arising
  from such a $\sigma$ will satisfy the polynomials of the theorem for every
  binary resolution. Thus in the statement of the theorem the phrase
  `if $\mathcal A$ is a clade on 
  $\sigma$' 
 can be replaced with `if
  $\mathcal A$ is a clade on a binary resolution of 
  $\sigma$.'

\begin{proof}[Proof of Theorem \ref{thm:cladesplitinv}] 

We derive this result in part using ideas developed for
  the construction of invariants for clade probabilities in
  \cite{adr2011b}.
  
  Let $Cl(\mathcal A)$ represent the event  that $\mathcal A_g$
is a clade on an observed gene tree.
Then note that $$\PP_\sigma(Sp(\mathcal A))=\PP_\sigma(Cl(\mathcal A))+\PP_\sigma(Cl(\bar{\mathcal A}))-\PP_\sigma(Cl(\mathcal A),Cl(\bar{\mathcal A})).$$
Thus equation \eqref{eq:cladesplitinv} will follow from establishing the three equalities:
\begin{align}
\sum_{\mathcal S\subseteq \mathcal A'} \PP_\sigma(Cl(\mathcal S\cup\{a\}\cup \mathcal C))
 & -  \sum_{\mathcal S\subseteq \mathcal A'}\PP_\sigma(Cl(\mathcal S\cup\{b\}\cup \mathcal C))=0,\label{eq:p1}\\
\sum_{\mathcal S\subseteq \mathcal A'} \PP_\sigma(Cl(\overline{\mathcal S\cup\{a\}\cup \mathcal C}))
 & - \sum_{\mathcal S\subseteq \mathcal A'}\PP_\sigma(Cl(\overline{\mathcal S\cup\{b\}\cup \mathcal C}))=0,\label{eq:p2}\\
 \sum_{\mathcal S\subseteq \mathcal A'} \PP_\sigma(Cl({\mathcal S\cup\{a\}\cup \mathcal C})&,Cl(\overline{\mathcal S\cup\{a\}\cup \mathcal C}))\label{eq:p3}\\
  -  \sum_{\mathcal S\subseteq \mathcal A'}&\PP_\sigma(Cl({\mathcal S\cup\{b\}\cup {\mathcal C}})),Cl(\overline{\mathcal S\cup\{b\}\cup {\mathcal C}}))=0,\notag
\end{align}

That equation \eqref{eq:p1} holds is Theorem 6 of \cite{adr2011b}.  To
establish equation \eqref{eq:p2}, for any $\mathcal S\subseteq\mathcal
A'$, let $\tilde{\mathcal S}= \mathcal A'\smallsetminus \mathcal S$,
and $\tilde {\mathcal C}=(\overline{\mathcal A})\smallsetminus
\mathcal C$.  With this notation
\begin{align*}
\overline{\mathcal S\cup\{a\}\cup \mathcal C}&=\tilde {\mathcal S}\cup\{b\}\cup \tilde{\mathcal C},
\end{align*}
where $\tilde{\mathcal S}\subseteq \mathcal A'$ and $\emptyset\ne
\tilde{\mathcal C}\subsetneq \mathcal X\smallsetminus \mathcal A$.  We
thus see equation \eqref{eq:p2} is another instance of the equation
\eqref{eq:p1}.  (It is essential here that $\mathcal C$ be a proper
subset of $\overline{\mathcal A}$, so that $\tilde {\mathcal C}$ is
nonempty; this is why $\mathcal A$ must exclude at least 2 taxa.)

Establishing equation \eqref{eq:p3} requires more argument. Using the
above notation, it can be restated as
\begin{multline*}
 \sum_{\mathcal S\subseteq \mathcal A'} \PP_\sigma(Cl({\mathcal S\cup\{a\}\cup \mathcal C}),Cl(\tilde{\mathcal S}\cup\{b\}\cup \tilde{\mathcal C}))\\
  - \sum_{\mathcal S\subseteq \mathcal A'}\PP_\sigma(Cl({\mathcal S\cup\{b\}\cup {\mathcal C}}),Cl(\tilde{\mathcal S}\cup\{a\}\cup\tilde{\mathcal C}))=0.
 \end{multline*} 
 We establish this by showing a version of it conditioned on the
 partition of $\mathcal A$ corresponding to lineages present at the
 $\MRCA(\mathcal A)$ in the species tree. Now consider any realization
 of the remainder of the coalescent process (\emph{i.e.}, on $\sigma$
 with edges below $\MRCA(\mathcal A)$ removed) which displays the
 clades ${\mathcal S_g\cup\{A\}\cup {\mathcal C}_g}$ and
 $\tilde{\mathcal S}_g\cup\{B\}\cup\tilde{\mathcal C}_g$. First note
 that $A$ and $B$ are in different partition sets at $\MRCA(\mathcal
 A)$, or these clades could not be formed.  But then this realization
 has the same probability as one where the lineages for $B$ and $A$
 are exchanged. This exchange leads to a gene tree displaying clades
 ${\hat{\mathcal S}_g\cup\{B\}\cup {\mathcal C}_g}$ and
 $\tilde{\hat{\mathcal S}}_g\cup\{A\}\cup\tilde{\mathcal C}_g$. Here
 $\hat{\mathcal S}$ is still a subset of $\mathcal A'$, but generally
 differs from $\mathcal S$ because some elements of $\mathcal S$ are
 in the partition sets with $A$ and $B$ at $\MRCA(\mathcal A)$.
 Conditioned on the partition, this establishes a bijective
 correspondence between equiprobable realizations of the coalescent
 contributing to the two sums in the equality, and thus the
 conditioned equality holds. Summing over all possible partitions of
 $\mathcal A$, weighted by their probabilities, give the unconditioned
 equation \eqref{eq:p3}.
\end{proof}

 \begin{ex}
   Here we explicitly give the clade-induced split invariants of
   Theorem \ref{thm:cladesplitinv} for 5-taxon species trees, and
   compare them to the full set of invariants for such trees.
 
 \smallskip
\noindent
\emph{Caterpillar tree} $((((a,b),c),d),e)$: Consider the 2-clade
$\mathcal A=\{a,b\}$, so $\mathcal A'=\emptyset$. With $\mathcal
C=\{c\}$, equation \eqref{eq:cladesplitinv} becomes
  $$\PP_\sigma (Sp(ac)) - \PP_\sigma(Sp(bc))=0.$$ 
For other choices of singleton $\mathcal C$ we find
   $$\PP_\sigma (Sp(ad)) - \PP_\sigma(Sp(bd))=0,$$
    $$\PP_\sigma (Sp(ae)) - \PP_\sigma(Sp(be))=0.$$  
    We refer to these as \emph{cherry-swapping split invariants},
    since they hold because $a,b$ form a 2-clade, and their lineages
    are thus exchangeable under the coalescent model.  Two-element
    choices of $\mathcal C$ give the same equalities, up to sign, as
    the ones already listed.

    For the clade $\{a,b,c\}$ using $a$ and $b$ as the two singleton
    taxa we find
  $$\PP_\sigma (Sp(ad))+\PP_\sigma (Sp(acd))   - \PP_\sigma(Sp(bd))-\PP_\sigma(Sp(bcd))=0,$$
  $$\PP_\sigma (Sp(ae))+\PP_\sigma (Sp(ace))   - \PP_\sigma(Sp(be))-\PP_\sigma(Sp(bce))=0,$$
both of which were already implied by the cherry-swapping invariants.
But using $a$ and $c$ as the single taxa we get
 $$\PP_\sigma (Sp(ad))+\PP_\sigma(Sp(abd)) - \PP_\sigma(Sp(cd))-\PP_\sigma(Sp(bcd))=0,$$
 $$\PP_\sigma (Sp(ae))+\PP_\sigma(Sp(abe)) - \PP_\sigma(Sp(ce))-\PP_\sigma(Sp(bce))=0.$$
 However, these are the same, up to a sign.  There is an additional
 invariant with $a$ and $b$ interchanged from this last one, which is
 obtained with $b,c$ chosen as the singletons. Alternately, it follows
 from the last one using the ``cherry-swapping'' exchangeability of
 lineages for $a$ and $b$.
  
 A computation with the algebra software {\tt Singular}
 shows these span the space of all linear split invariants for this
 rooted tree. Note that the tree $((((a,b),c),e),d)$ would produce
 exactly the same set of invariants, so by evaluating linear
 invariants one would not be able to identify the root of such a caterpillar tree. This is an
 instance of Corollary \ref{thm:idroot} (a) below.
 
 \medskip \emph{Balanced tree} $(((a,b),c),(d,e))$: This tree has all
 the clades excluding at least 2 taxa that the caterpillar does, but
 in addition displays $\{d,e\}$.  Thus all the invariants listed for
 the caterpillar hold, as well as additional ones from this cherry.
 For instance,
$$\PP_\sigma (Sp(ad)) - \PP_\sigma(Sp(ae))=0.$$
These span the space of linear split invariants for this tree, as
computed by {\tt Singular}.

\medskip

\emph{Pseudocaterpillar} $(((a,b),(d,e)),c)$: This tree has only two
clades that exclude at least two taxa, namely $\{a,b\}$ and $\{d,e\}$
From the first of these clades we obtain the invariants
$$\PP_\sigma (Sp(ac)) - \PP_\sigma(Sp(bc))=0,$$ 
   $$\PP_\sigma (Sp(ad)) - \PP_\sigma(Sp(bd))=0,$$
    $$\PP_\sigma (Sp(ae)) - \PP_\sigma(Sp(be))=0,$$ 
and for the second
$$\PP_\sigma (Sp(ad)) - \PP_\sigma(Sp(ae))=0,$$ 
   $$\PP_\sigma (Sp(bd)) - \PP_\sigma(Sp(be))=0,$$
    $$\PP_\sigma (Sp(cd)) - \PP_\sigma(Sp(ce))=0.$$
Note that of these six invariants, the middle four are linearly dependent, 
with a 3-dimensional span.  The span of all six invariants is 5-dimensional. 
All can be explained by cherry-swapping.
    
For the 5-taxon pseudocaterpilar, a computation with {\tt Singular} produces
a 6-dimensional space of linear invariants, with
\begin{equation}\label{eq:5pcextra}
\PP_\sigma(Sp(ab))+2\PP_\sigma(Sp(bc))-2\PP_\sigma(Sp(ce))-\PP_\sigma(Sp(de))
\end{equation}
as the additional generator. Note that neither of Theorems
\ref{thm:4ptinv} and \ref{thm:cladesplitinv} give equalities involving
$\PP_\sigma(Sp(ab))$ or $\PP_\sigma(Sp(de))$ for this tree, so they
cannot provide an explanation for this invariant. One will be given in
Proposition \ref{prop:special5} below.

\medskip

For the 5-taxon trees, our {\tt Singular} computations found a Gr\"obner
basis for all invariants in split probabilities.  For the
pseudocaterpillar, there were only linear polynomials, indicating that
the ones above imply all higher degree invariants. For the caterpillar
and balanced trees there were additional non-linear invariants in the
basis.  These are given in Appendix \ref{app:grobner}. However, we
have no theoretical understanding of them.
\end{ex}

\begin{ex} For each of the six shapes for 6-taxon species trees, we
  computed invariants for the split probabilities using {\tt Singular}. In
  order to make the computations terminate, we limited the degree to
  8. For all shapes we found that the clade-induced split invariants
  given by Theorem \ref{thm:cladesplitinv} spanned the space of linear
  invariants; that is, there were no `extra' linear invariants such as
  that found for the 5-taxon pseudocaterpillar species trees.  Table
  \ref{Table:6dim} shows the dimension of the space of non-trivial
  linear invariants, which depends upon the rooted topology.

  For the 6-taxon trees we found no non-linear invariants of degree
  less than our bound. However a dimension argument indicates
  higher-degree invariants must exist: There are 25 non-trivial split
  probabilities. After accounting for the trivial split invariant of
  equation \eqref{eq:trivialinv}, the space the linear invariants
  define is of dimension 24 minus the dimension shown in Table
  \ref{Table:6dim}. Since the variety of split probabilities has
  dimension at most 4 (the number of internal edges on the species
  tree), higher degree invariants must exist.

\begin{table}
\begin{tabular}{|c|c|}
\hline
Species Tree & dim\\
\hline
$(((((a,b),c),d),e),f)$ & 11\\
$((((a,b),(c,d)),e),f)$&13\\
$((((a,b),c),(d,e)),f)$ & 14 \\
$((((a,b),c),d),(e,f))$ & 14\\
$(((a,b),(c,d)),(e,f))$ & 15\\
$(((a,b),c),(d,(e,f)))$ &15\\
\hline
\end{tabular}
\caption{Dimensions of the space of non-trivial linear invariants for 6-taxon species trees. These are defined on a space of
  $(2^5-6 - 1) =25$ non-trivial split probabilities. }\label{Table:6dim}
\end{table}
\end{ex}

\section{Identifiability of the rooted species tree from split probabilities}\label{sec:ident}

We first show that the clade-induced split invariants of the last
section, which vanish if a species tree has a particular clade, do not
vanish for generic parameter choices if the species tree lacks that
clade (with some exceptions). This is the main ingredient in obtaining
Theorem \ref{thm:idroot}, that the rooted topological species tree is
recoverable from split probabilities in most circumstances.

The following lemma is key to our argument.

\begin{lemma}\label{lem:nonzero} 
  Let $\psi$ be a binary rooted topological species tree on $\mathcal X$, and $\mathcal
  X=\mathcal A\sqcup \mathcal D$ a disjoint union of 
  subsets with $|\mathcal A|, |\mathcal D| \ge 2$. 
  Suppose 
  \begin{enumerate}
\item  $\mathcal A$ is not a clade on $\psi$, 
\item $\mathcal D$ is not a 2-clade on $\psi$. \label{cond:notcherry}
\end{enumerate}
Then there exists some
  $\emptyset\neq\mathcal C\subsetneq\mathcal D$, $a,b\in\mathcal A$, and some choice
  of edge lengths $\lambda$ on $\psi$ such that the clade-induced split
  invariant of Theorem \ref{thm:cladesplitinv}, equation \eqref{eq:cladesplitinv} does not vanish on the split
  probabilities arising under the multispecies coalescent model on
  $\sigma=(\psi,\lambda)$.
\end{lemma}

\begin{proof}
We consider two cases, according to whether or not $\mathcal D$ is a clade displayed on $\psi$.

\smallskip

First suppose $\mathcal D$ is not a clade on $\psi$.  Pick a minimal
clade displayed on $\psi$ that contains at least one element of
$\mathcal A$ and at least one element of $\mathcal D$, and let $v$ be
its MRCA. Let $w_1$ and $w_2$ be the children of $v$. Note that the
minimality of the clade implies one of the $w_i$, say $w_1$, has as
its leaf descendants only elements of $\mathcal A$, and the other, say
$w_2$, has as its leaf descendants only elements of $\mathcal D$.
Also observe that $\mathcal A\smallsetminus \desc(v)$ is nonempty,
since otherwise the leaf descendants of $w_1$ would have to be all of
$\mathcal A$, contradicting that $\mathcal A$ is not a clade on the
tree. Similarly, $\mathcal D\smallsetminus \desc(v)$ is nonempty.
These statments furthermore imply $v$ is not the root of the tree, so
it has a parent $u$.

Choose $a\in\desc(w_1)\subsetneq \mathcal A$, $b\in \mathcal
A\smallsetminus\desc(v)$, and $\emptyset\neq\mathcal
C=\desc(w_2)\subsetneq \mathcal D$.  Let all edge lengths on $\psi$
below $v$ have length (near) zero, edge $(u,v)$ have length (near)
infinity, and the remaining edges have any finite positive length.
Then a gene tree arising from the coalescent model will have (near)
zero probability of displaying $Sp(\mathcal S\cup\{b\}\cup \mathcal
C)$ for $\mathcal S\subseteq \mathcal A'$, since any displayed split
(of non-negligible probability) with $\{B\}\cup\mathcal C_g$ in a
partition set is (near) certain to contain all of $\desc(v)$, and
hence $A$, in that set as well.  Thus all negative terms in equation
\eqref{eq:cladesplitinv} are negligible.  On the other hand there is a
positive term in that equation for $\PP(Sp(\desc(v)))$, which has
value (near) 1. Thus equation \eqref{eq:cladesplitinv} does not hold.

\smallskip

Next we consider the case when $\mathcal D$ is a clade displayed on
$\psi$, but, by condition \eqref{cond:notcherry}, $\mathcal D$ has at
least 3 elements.
 
We first consider a particular form for $\psi$, and will then reduce
the general tree to this form. To this end, suppose $\mathcal
D=\{d_1,d_2,d_3\}$ and $\psi $ is the rooted caterpillar tree
$$(((\dots((((d_1,d_2),d_3),a),c_1)\dots), c_n),b)$$ 
with at least 5 taxa, $\mathcal A=\{a,b,c_1,\dots, c_n\}$ with $n\ge
0$, $a,b$ chosen as shown. Let $w=\MRCA(\mathcal D)$, and $v$ its
parent, so $v$ is also the parent of $a\in \mathcal A$. Let $\mathcal
C=\{d_3\}$. Choose all internal edge lengths of $\psi$ to be (near)
infinite, except for those below $v$ which we choose to be (near)
zero. Then all rooted gene trees realizable with non-negligable
probability will be formed by $\mathcal D_g\cup\{A\}$ coalescing into
some rooted gene tree in the branch above $v$, with this subtree then
joining to the remaining taxa in $\mathcal A$ in a tree that otherwise
exactly matches the caterpillar structure of $\psi$.  Thus the event
$Sp(S\cup\{a\}\cup\mathcal C)$ is non-negligibly realizable only for
$\mathcal S=\emptyset$, by gene trees with the rooted subtree on
$\mathcal D_g\cup\{A\}$ having $\{D_3,
A\}$ as a clade. But the probability of this 2-clade forming is
(near) $4/18=2/9$, since of the 18 ranked rooted trees on 4 taxa, four
have any given 2-clade.  On the other hand, the event
$Sp(S\cup\{b\}\cup\mathcal C)$ is non-negligibly realizable only for
$S=\{c_1,\dots c_n\}$, by gene trees where the 4-taxon rooted subtree
on $\mathcal D_g\cup\{A\}$ has $D_3$ as an outgroup. Such gene trees
occur with probability $3/18=1/6$. Thus the invariant of equation
\eqref{eq:cladesplitinv} evaluates (near) to $2/9-1/6$, and is thus
not zero.

Now for the general case, in which $\mathcal D$ is a clade with 3 or
more elements, by picking some internal edges of $\psi$ within the
subtree on $\mathcal D$ to have (near) infinite and zero lengths we
can ensure that with probability (near) 1 that $\mathcal D$ coalesces
into exactly 3 lineages by $\MRCA(\mathcal D)$. Similarly by picking
(near) infinite edge lengths for those edges leading off of the path
between $\MRCA(\mathcal D)$ and the root of $\psi$ to groups of
elements in $\mathcal A$, we can ensure with probability 1 that these
groups have coalesced before reaching that path. Then the argument
above for the caterpillar tree applies with lineages for groups of
taxa replacing the individual ones.
 \end{proof}

 That condition \eqref{cond:notcherry} of the above lemma is necessary
 is shown by the following.
\begin{prop}
Let $\psi$ be a species tree topology on $\mathcal X$, and $\mathcal
  X=\mathcal A\sqcup \mathcal D$ a disjoint union of 
  subsets with $|\mathcal D| = 2$. Then if $\mathcal D$ is a clade on $\psi$, 
  the 
  polynomials defined for
 $\mathcal A$ 
by equation  \eqref{eq:cladesplitinv}
   in Theorem \ref{thm:cladesplitinv} all vanish, 
  regardless of whether $\mathcal A$ is a clade on $\psi$.
\end{prop}
\begin{proof}
  Since $\mathcal D=\{d_1,d_2\}$, the clade-induced split invariants for
  $\mathcal A$ in equation \eqref{eq:cladesplitinv} require that
  $\mathcal C$ be a singelton set, which we may assume is $\mathcal
  C=\{d_1\}$.  Since $\mathcal D$ is a 2-clade, by exchangeability of
  lineages under the coalescent implies
\begin{align*}\PP_\sigma(Sp(\mathcal S\cup\{a\}\cup \mathcal \{d_1\}))&=\PP_\sigma(Sp(\mathcal S\cup\{a\}\cup \mathcal \{d_2\}))\\
&=\PP_\sigma(Sp(\tilde{\mathcal S}\cup\{b\}\cup \mathcal \{d_1\})),
\end{align*}
where the last equality is obtained by taking the complementary split set.
Thus equation \eqref{eq:cladesplitinv} holds, since terms cancel in pairs.
\end{proof}

From Theorem \ref{thm:cladesplitinv} we obtain the following. 

\begin{cor}\label{thm:meas0}
  Let $\psi$ be a rooted binary species tree topology on at least 5
  taxa $\mathcal X$, where $\mathcal X=\mathcal A\sqcup \mathcal D$ is
  a disjoint union of subsets with $|\mathcal A|, |\mathcal D| \ge 2$.
  If $\mathcal A$ is not a clade on $\psi$ and $\mathcal D$ is not a
  2-clade, then for generic choices of internal edge lengths $\lambda$
  (i.e., all except those in some set of measure zero) there exists
  some $\emptyset\neq\mathcal C\subsetneq\mathcal D$, $a,b\in\mathcal
  A$, such that the corresponding clade-induced split invariant of
  equation \eqref{eq:cladesplitinv} does not vanish on the split
  probabilities arising under the multispecies coalescent on
  $\sigma=(\psi,\lambda)$.
\end{cor}

\begin{proof} The non-trivial split probabilities arising from the
  coalescent on $\sigma$ can be expressed as polynomials in the $\exp
  (-\lambda_i)$. We can thus view the set of all vectors of split
  probabilities as the image of $(0,1)^{n-2}$ under a polynomial map,
  which is therefore a semi-algebraic set.  By Lemma
  \ref{lem:nonzero}, there is an invariant which does not vanish at
  some point in this set, so the composition of the invariant with the
  polynomial map is not identically zero on $(0,1)^{n-2}$.  Since this
  composition is a polynomial, its non-vanishing at a single point
  implies the set where it vanishes has measure zero in $(0,1)^{n-2}$.
  Mapping this set to interior edge lengths by $x=-\log(X)$
  shows the set of edge lengths
  for which the invariant vanishes has measure zero.
\end{proof}

\begin{cor}\label{cor:genid}
  Let $\psi$ be a rooted binary species tree topology on a set
  $\mathcal X$ of at least 5 taxa.  For generic edge lengths
  $\lambda$, all clades on $\psi$ excluding at least three taxa can be
  identified by evaluating clade-induced split invariants on the
  probabilities of splits under the multispecies coalescent on
  $\sigma=(\psi,\lambda)$.
  
  Clades on $\psi$ excluding exactly two taxa can similarly be identified 
  if their complement is not a 2-clade.
\end{cor}

\begin{proof}
  For any subset of $\mathcal A\subsetneq \mathcal X$ excluding at
  least three taxa, if we find any invariant given by Theorem
  \ref{thm:cladesplitinv} that fails to vanish on the split
  probabilities for $\sigma=(\psi,\lambda)$, then $\mathcal A$ is not
  a clade on $\psi$. If all such invariants vanish, then by Corollary
  \ref{thm:meas0}, either $\mathcal A$ is a clade on $\psi$, or
  $\lambda$ lies in a set of measure zero (which is dependent on
  $\mathcal A, \mathcal C,a,$ and $b$ used in defining the invariant).

  Thus, considering all such $\mathcal A$, we can determine all clades
  excluding at least three taxa, unless the edge lengths $\lambda$ lie
  in a set of measure zero (the finite union of sets of measure zero
  for each clade, each of which is a finite intersection of sets of
  measure zero for each invariant for that clade.)

Finally, suppose
  $\mathcal A$ excludes only two taxa, with complement $\mathcal D$. Then for generic edge lengths the non-vanishing of an appropriate
  invariant can detect whether $\mathcal D$ is a 2-clade. If it is not, then using
  this knowledge, the vanishing of all clade split invariants
  associated to $\mathcal A$ will identify it as a clade.
  \end{proof}

We now are able to use split invariants to fully identify rooted species trees in some cases, and find only 2 or 3 possible rootings in others. Although this result will be strengthened in Theorem \ref{thm:rootID} below by also using some inequalities, equalities alone lead to the following result.

\begin{cor}\label{thm:idroot}
  A binary rooted species tree topology can be identified from split
  probabilities via the clade-induced split invariants of equation
  \eqref{eq:cladesplitinv} for generic edge lengths on all species
  trees on 5 or more taxa, except in the following cases of
  indeterminacy. Here $T$ denotes a rooted subtree on 3 or more taxa,
  which is identifiable, and lower case letters denote other taxa.
\begin{enumerate}
\item \label{case:a}$((T,a),b)$, $((T,b),a)$
\item\label{case:b} $((T,a),(b_1,b_2))$, $((T,(b_1,b_2)),a)$
\item\label{case:c}$((T,(a_1,a_2)),(b_1,b_2))$, $((T,(b_1,b_2)),(a_1,a_2))$,
\item \label{case:d} $(((a,b),c),(d,e))$, $((a,b),(c,(d,e)))$, $((a,b),(d,e)),c)$ 
\item \label{case:e}$(((a,b),(c,d)),(e,f))$, $(((a,b),(e,f)),(c,d))$, $(((c,d),(e,f)),(a,b))$
\end{enumerate}
\end{cor}
The various cases enumerated in the corollary are depicted in Figure \ref{fig:specialcases}.

\begin{center}
\begin{figure}
\includegraphics[width=12cm]{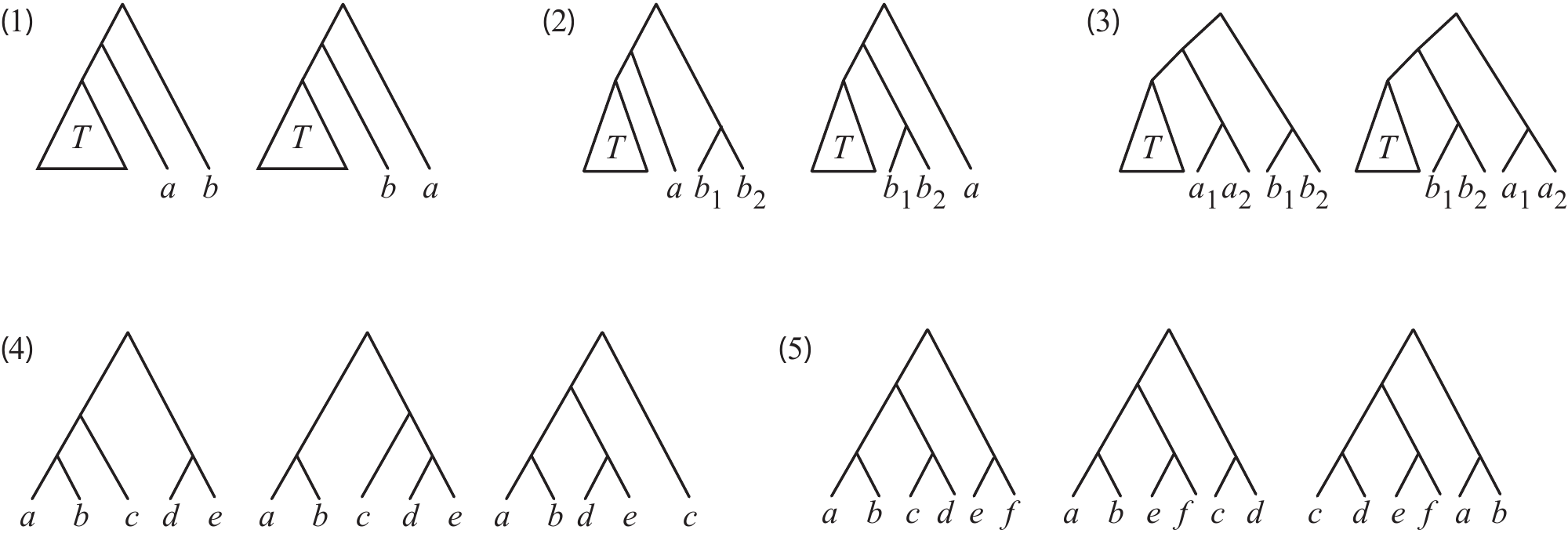}
\caption{For generic species tree edge lengths, split invariants can be used to determine rooted 
species tree topologies, up to the 5 ambiguous cases shown here, as proved in 
Corollary \ref{thm:idroot}.} \label{fig:specialcases}
\end{figure}
\end{center}

\begin{proof} Given all split probabilities 
computed from a species trees with generic edge lengths, we may test every subset 
of $\mathcal X$ omitting 3 or more taxa to see if  it is a clade on $\psi$,
using Corollary \ref{cor:genid}.
We then form a list of all such clades (including trivial ones) on $\psi$.

If two clades on this list form a bipartition of $\mathcal X$, then we have 
determined all clades on $\psi$, hence its rooted topology. 

If no pair of clades on this list partition $\mathcal X$, but we find
three clades on the list that do, denote them by $\mathcal A$,
$\mathcal B$, and $\mathcal C$ with $|\mathcal A|\ge| \mathcal B|\ge
|\mathcal C|$.  Since there are at least 5 taxa, we cannot have
$|\mathcal A|=1$. If $|\mathcal A|=2$ then $|\mathcal B|=2$,
$|\mathcal C|=1$ or $2$, yielding cases \ref{case:d} and
\ref{case:e}.  If $|\mathcal A|\ge 3$, then we know $\mathcal
B\cup\mathcal C$ is not a clade, since otherwise it would have
appeared on the list, leading to a bipartition of $\mathcal X$.  Thus
either $\mathcal A\cup \mathcal B$ or $\mathcal A\cup \mathcal C$ is a
clade. Note that $|\mathcal B|\neq 1$, else $\mathcal A$ would not
omit at least 3 taxa. If $|\mathcal B|=2$, then we obtain cases
\ref{case:b}, and \ref{case:c}.  If $|\mathcal B|\ge3$, then
$\mathcal A\cup \mathcal B$ is a clade, since $\mathcal A\cup\mathcal
C$ and $\mathcal B\cup\mathcal C$ were not found on the list. Since
all subclades of $\mathcal A$, $\mathcal B$, and $\mathcal C$ appear
in the list, all clades on $\psi$ are determined.

If there is no partition $\mathcal X$ into two or three clades on the
list, then there must be one with four, since if five or more were
needed then the union of each pair would omit at least 3 taxa and at
least one such union is a detectable clade. Denote the four clades by
$\mathcal A$, $\mathcal B$, $\mathcal C$, and $\mathcal D$, with
$|\mathcal A|\ge| \mathcal B|\ge |\mathcal C|\ge |\mathcal D|$. We
also must have $|\mathcal C|=|\mathcal D|=1$, since otherwise the
union of each pair of sets would omit at least 3 taxa, and hence would
already have been tested for being a clade. It is now enough to
determine the clade structure formed by the union of these sets, since
all subclades of them are already known.

If $|\mathcal B|=1$, then $|\mathcal A|\ge 2$, and none of $\mathcal
B\cup\mathcal C$, $\mathcal B\cup \mathcal D$, and $\mathcal
C\cup\mathcal D$ are clades since they omit at least 3 taxa and did
not appear on the list of known clades.  Thus the four clades must
form a rooted unbalanced 4-leaf tree with $\mathcal A$ in the
cherry. We can then use invariants to check which of $\mathcal
A\cup\mathcal B$, $\mathcal A\cup\mathcal C$, $\mathcal A\cup\mathcal
D$ is a clade, since we know their complement is not a 2-clade. This
results in case \ref{case:a}.

If $|\mathcal B|\ge 2$, then every pairwise union of the four clades
except $\mathcal A\cup\mathcal B$ would have been tested, so $\mathcal
A\cup\mathcal B$ must be a clade. As $\mathcal C\cup\mathcal D$ is not
a clade, this also falls into case \ref{case:a}.
\end{proof}

The identifiability results of Corollary \ref{thm:idroot} are based
solely on the use of clade-induced linear invariants, that is on
certain linear equalities. Fortuitously, by considering other linear
invariants and inequalities, we can strengthen these results. For
example, those trees in case \ref{case:d} of Corollary
\ref{thm:idroot} can be distinguished by considering the sign (+/-) or
vanishing (= 0) of the linear invariant given below.  Indeed, this
linear expression in split invariants is equivalent to that of the
computationally-determined expression \eqref{eq:5pcextra}, and the
following proposition gives a theoretical justification for its
existence.  However, this linear invariant appears to be a special one
for a single 5-taxon tree, with no analogs for other trees.

\begin{prop} \label{prop:special5} The expression
\begin{multline}\label{eq:5inv}
\PP_\sigma(Sp(ab))+\PP_\sigma(Sp(ac))+\PP_\sigma(Sp(bc))\\-\PP_\sigma(Sp(de))-\PP_\sigma(Sp(cd))-\PP_\sigma(Sp(ce))
\end{multline} 
evaluates to 0 for the species tree $(((a,b),(d,e)),c)$. Assuming all
species tree edge lengths are finite and positive, expression
\eqref{eq:5inv} is positive for the species tree $(((a,b),c),(d,e))$
and negative for the species tree $((a,b),(c,(d,e)))$.
\end{prop}

\begin{proof}
  For the species tree $(((a,b),(d,e)),c)$, let $e_1$ be the edge
  immediately above $\MRCA(a,b)$, and $e_2$ the edge above
  $\MRCA(d,e)$ in the species tree.  To show expression
  \eqref{eq:5inv} evaluates to 0, it is enough to show this conditioned on disjoint and
  exhaustive events.  To this end, we compute \eqref{eq:5inv} conditioned on
 whether coalescent events occur on edges $e_1$ and
  $e_2$.

Given that no coalescence occurs on either
  $e_1$ or $e_2$, the probabilities of $Sp(ab)$ and $Sp(de)$ are
  equal by exchangeability. Similarly, the other 4 split probabilities
  appearing in formula \eqref{eq:5inv} are all equal. Thus all terms
  cancel.

  Given that coalescences occurs on both $e_1$ and $e_2$, then the
  probabilities of $Sp(ab)$ and $Sp(de)$ are both 1. The other 4
  probabilities are all 0, so again all terms cancel.

  Assuming that a coalescent event occured on exactly one of $e_1$ and
  $e_2$, without loss of generality we may assume it is on $e_1$. Then
  $Sp(ab)$ has probability 1, while $Sp(ac)$ and $Sp(bc)$ have
  probability 0.  The next coalescent event produces the only other
  non-trivial split of the gene tree, which must be one of $Sp(CD),
Sp(CE), Sp(DE)$.  Thus $\PP(Sp(cd))+ \PP(Sp(ce))+\PP( Sp(de))=1$,
  and again we find the expression gives 0.

\medskip

For the species tree $(((a,b),c),(d,e))$, let $e_1$ be the edge above
$\MRCA(a,b)$, $e_2$ that above $\MRCA(a,b,c)$, and $e_3$ that above
$\MRCA(d,e)$. We will again consider disjoint exhaustive events, and
show that conditioned on the number of coalescent events on
these edges expression \eqref{eq:5inv} is always non-negative, and
sometimes positive. Thus, the unconditioned expression is positive.

If there are exactly 2 coalescences on the edges $e_1,e_2$,
then in the formation of a gene tree a compound lineage $ABC$ enters the
population above the root, and exactly one of  $Sp(AB), Sp(AC),
Sp(BC)$ form.  Moreover, $Sp(DE)$ will be present on any unrooted
version of such a gene tree, and $Sp(CD), Sp(CE)$ absent.  Thus, in
the conditional probability, the first three terms of expression
\eqref{eq:5inv} sum to 1, and the last three terms to $-1$, for a
total of 0.

If there is exactly 1 coalescence on the edges $e_1, e_2$,
then again exactly one of $Sp(AB), Sp(AC), Sp(BC)$ must form on
a gene tree, and the sum of the first three probabilities in
\eqref{eq:5inv} is 1.  If $Sp(AB)$ formed, then exactly one of $Sp(DE),Sp(CD),Sp(CE)$ forms, and the expression
in \eqref{eq:5inv} is zero. If $Sp(AB)$ does not form, say instead
$Sp(AC)$ does, then neither $Sp(CD)$ nor $Sp(CE)$ can appear on any such
gene tree, while $Sp(DE)$ forms with probability less than 1 since
$e_3$ has finite length.  In this case, expression
\eqref{eq:5inv} is positive. Similarly, if $Sp(BC)$ forms, then the
expression is positive.

If there are no coalescences on $e_1$, $e_2$, or $e_3$, then all 5
lineages of the taxa arrive at the root of the species tree
distinct. Then by exchangeability one sees the probability of every
split $Sp(xy)$ is the same, so the expression evaluates to 0.

If there are no coalescences on $e_1$, $e_2$, but there is one on
$e_3$, then $Sp(DE)$ forms, but not $Sp(CD)$ nor $Sp(CE)$.  Thus the last
three terms yield $-1$. As the lineages at the root of the species
tree will be $A$, $B$, $C$, and a combined $DE$, exactly one of 
$Sp(AB)$, $Sp(AC)$, $Sp(BC)$ forms, so the first three terms add to 1, and
all terms cancel.

\medskip

The claim for the species tree $((a,b),(c,(d,e)))$ follows by interchanging 
taxon names from $(((a,b),c),(d,e))$.

\end{proof}

To address root identifiability in case \ref{case:a} of Corollary
\ref{thm:idroot}, we have the following.

\begin{prop} \label{prop:acbc} Let $\mathcal X$ be a set of at least 5
  taxa, $a,b\in \mathcal X$, and $T$ any rooted species tree topology
  on $\mathcal X'=\mathcal X\smallsetminus \{a,b\}$. Let $c\in\mathcal
  X'$. Suppose $\psi$ is one of the species trees $((T,a),b)$,
  $((T,b),a)$, or $(T,(a,b))$, and $\sigma=(\psi,\lambda)$ has
  positive length edges incident to the root. Then
$$\PP(Sp(ac))- \PP(Sp(bc))\begin{cases} > 0 \text{ if and only if } \psi =((T,a),b),\\
  =0\text{ if and only if } \psi =(T,(a,b)),\\
  <0 \text{ if and only if } \psi =((T,b),a).
\end{cases} $$
 \end{prop}

The intuition behind this proposition is rather simple. 
 The polynomial $\PP(Sp(ac))- \PP(Sp(bc))$ is a clade-induced split
 invariant for the tree $(T,(a,b))$, identifying by its vanishing that $ab$ is a clade.
 One might reasonably hope that the hyperplane defined by this
 invariant's vanishing separates collections of split probabilities
 for the two alternative trees $((T,a),b))$ and $((T,b), a))$. That this is true is established 
by a rather technical proof which appears in Appendix \ref{app:proofs}.

 For case \ref{case:b}, we follow a similar tack, focusing on a split
 invariant for the clade $ab_1b_2$ on the tree $(T,(a,(b_1,b_2))$. The
 proof of the following is also in Appendix \ref{app:proofs}.

 \begin{prop} \label{prop:abbc} Let $\mathcal X$ be a set of at least
   6 taxa, $a,b_1,b_2\in \mathcal X$, and $T$ any rooted species tree
   topology on $\mathcal X'=\mathcal X\smallsetminus
   \{a,b_1,b_2\}$. Let $c\in\mathcal X'$. Suppose $\psi$ is one of the
   species trees $((T,a),(b_1,b_2))$, $((T,(b_1,b_2)),a)$, or
   $(T,(a,(b_1,b_2)))$, and $\sigma=(\psi,\lambda)$ has positive
   length edges incident to the root. Then
\begin{multline*}
\PP(Sp(ac))+\PP(Sp(ab_2c))
-\PP(Sp(b_1c))- \PP(Sp(b_1b_2c))\\
\begin{cases} 
> 0 \text{ if and only if } \psi  =((T,a),(b_1,b_2))\\
=0\text{ if and only if } \psi  =(T,(a,(b_1,b_2)))\\
<0 \text{ if and only if } \psi  =((T,(b_1,b_2)),a)
\end{cases} 
\end{multline*}
\end{prop}

For case \ref{case:c}, we similarly have the following, also proved in the appendix.

\begin{prop} \label{prop:aabbc} Let $\mathcal X$ be a set of at least
  7 taxa, $a_1,a_2,b_1,b_2\in \mathcal X$,and $T$ any rooted species
  tree topology on $\mathcal X'=\mathcal X\smallsetminus
  \{a_1,a_2,b_1,b_2\}$.  Let $c\in\chi'$. Suppose $\psi$ is one of the
  species trees $((T,(a_1,a_2)),(b_1,b_2))$,
  $((T,(b_1,b_2)),(a_1,a_2))$, or $(T,((a_1,a_2),(b_1,b_2)))$, and
  $\sigma=(\psi,\lambda)$ has positive length edges incident to the
  root.  Then,
\begin{multline*}
\PP(Sp(a_1c))+\PP(Sp(a_1a_2c))+\PP(Sp(a_1b_2c))+\PP(Sp(a_1a_2b_2c))\\
-\PP(Sp(b_1c))- \PP(Sp(b_1a_2 c))- \PP(Sp(b_1 b_2 c))- \PP(Sp( b_1a_2 b_2 c))\\
\begin{cases} 
> 0 \text{ if and only if } \psi  =\big((T,(a_1,a_2)),\; (b_1,b_2) \big)\\
=0\text{ if and only if } \psi  = \big(T,\; ((a_1,a_2),(b_1,b_2)) \big)\\
<0 \text{ if and only if } \psi = \big((T,(b_1,b_2)),\; (a_1,a_2) \big)
\end{cases} 
\end{multline*}
\end{prop}

We summarize these results with the following.
\begin{thm}\label{thm:rootID}
  For any species tree on 5 or more taxa with generic edge lengths,
  the rooted species tree topology is identifiable from split
  probabilities by testing linear equalities and inequalities, with
  the possible exception of case \ref{case:e} of Theorem
  \ref{thm:idroot}, the 6-taxon rooted trees with three 2-clades.
\end{thm}

Note that we do not claim that there do not exist linear inequalities
that could be used to identify the root in case \ref{case:e}, only
that we have not found any among the candidates we considered for this
purpose.  Moreover, non-linear split invariants for those trees might
be useful for root identification, but they are of higher degree than
we were able to compute and remain unknown.  While the practical
import of this special case is small, understanding it better is
desirable nonetheless.

\section{Acknowledgements}  
This work was begun while ESA and JAR were Short-term Visitors and JHD
was a Sabbatical Fellow at the National Institute for Mathematical and
Biological Synthesis, an institute sponsored by the National Science
Foundation, the U.S. Department of Homeland Security, and the
U.S. Department of Agriculture through NSF Award \#EF-0832858, with
additional support from the University of Tennessee, Knoxville.  It
was further supported by the National Institutes of Health grant R01
GM117590, awarded under the Joint DMS/NIGMS Initiative to Support
Research at the Interface of the Biological and Mathematical Sciences.

\appendix
\section{Greedy split consensus on 5-taxon trees: proofs}\label{app:greedy}

Here we prove Propositions \ref{prop:greedynoncat} and
\ref{prop:greedycat} from Section \ref{sec:basic}.

With $\mathcal X=\{a,b,c,d,e\}$, there are 10 non-trivial splits, each
with blocks of size 2 and 3. We use the enumeration of splits and
their probabilities given in Example \ref{ex:5splits}. Computations,
assisted by the software COAL \cite{DegnanSalter2005}, produce the
formulas in Table \ref{tab:splitprobs} for these split probabilities
on the 3 species tree shapes, $\sigma_{bal}$, $\sigma_{pc}$, and
$\sigma_{cat}$.

\begin{table}$$\begin{array}{|c|c|c|}
\hline
& \sigma_{bal}&\sigma_{pc}\\
\hline
s_1& 1-\frac{2}{15} XY^3Z-\frac{2}{3}X &  1-\frac{1}{45}XYZ^6-\frac{1}{9}XY-\frac{2}{3}X \\
s_2,s_5& -\frac{2}{15} XY^3Z+\frac{1}{3}X & \frac{1}{30}XYZ^6-\frac{1}{6}XY+\frac{1}{3}X\\
s_3,s_4,s_6,s_7 &\frac{1}{30}XY^3Z+\frac{1}{6}XYZ & -\frac{1}{45}XYZ^6+\frac{2}{9}XY\\
s_8,s_9&\frac{1}{30}XY^3Z-\frac{1}{6}XYZ+\frac{1}{3}YZ& \frac{1}{30}XYZ^6-\frac{1}{6}XY+\frac{1}{3}Y\\
s_{10}&1+\frac{1}{5}XY^3Z-\frac{1}{3}XYZ-\frac{2}{3}YZ& 1-\frac{1}{45}XYZ^6-\frac{1}{9}XY-\frac{2}{3}Y\\
\hline
\end{array}
$$

$$\begin{array}{|c|c|}
\hline
&\sigma_{cat}\\
\hline
s_1 & 1-\frac{1}{45}XY^3Z^6-\frac{1}{9}XY^3-\frac{2}{3}X \\
s_2, s_5 & -\frac{1}{45}XY^3Z^6-\frac{1}{9}XY^3+\frac{1}{3}X \\
s_3, s_6 & -\frac{1}{45}XY^3Z^6+\frac{1}{18}XY^3+\frac{1}{6}XY \\
s_4, s_7 & \frac{1}{30}XY^3Z^6+\frac{1}{6}XY\\
s_8 & -\frac{1}{45}XY^3Z^6+\frac{1}{18}XY^3-\frac{1}{6}XY+\frac{1}{3}Y\\
s_9 & \frac{1}{30}XY^3Z^6-\frac{1}{6}XY+\frac{1}{3}Y\\
s_{10} & 1+\frac{1}{30}XY^3Z^6+\frac{1}{6}XY^3-\frac{1}{3}XY-\frac{2}{3}Y\\
\hline
\end{array}
$$
\caption{Split probabilities for gene trees arising on the 5-taxon species trees under the multispecies coalescent model.}\label{tab:splitprobs}
\end{table}

\begin{proof}[Proof of Proposition \ref{prop:greedynoncat}]
Given the equalities of split probabilities in Table \ref{tab:splitprobs} we need only show
that $s_1, s_{10}  \ge s_2, s_3, s_8$. Note that positive branch lengths imply $0<X,Y,Z<1.$

\smallskip

\noindent Then, for $\sigma_{bal}$, one finds
\begin{align*}s_1 - s_2 &= 1-X>0,\\
s_1 - s_3 &= 1 - \frac{1}{6} XY^3Z -\frac{2}{3}X - \frac{1}{6} XYZ>1-\frac 16-\frac 13-\frac 16=0,\\
s_1 - s_8 &= 1 - \frac{1}{6} XY^3Z -\frac{2}{3} X + \frac{1}{6} XYZ - \frac{1}{3} YZ \\
&= 1 + \frac{1}{6} XYZ(1-Y^2)  - \frac{2}{3} X  - \frac{1}{3} YZ\\
&>1+0-\frac23 -\frac 13=0,\\
s_{10} - s_2 &= 1 + \frac{1}{3} XY^3 Z -\frac{1}{3} XYZ - \frac{1}{3}X - \frac{2}{3} YZ \\
&= 1 + \frac {Z}3 \left (XY^3 - XY - 2Y\right) - \frac{1}{3}X>1-\frac 23-\frac 13=0\\ 
&\text{ since $XY^3 -XY - 2Y$ has minimum $-2$ on the unit square,}\\
s_{10} - s_3 &= 1 + \frac{1}{6} XY^3 Z - \frac{1}{2} XYZ - \frac{2}{3} YZ\\
&= 1 + \frac{1}{6} XZ (Y^3 - 3Y) - \frac{2}{3} YZ>1 + \frac{1}{6}(-2) - \frac{2}{3} =0\\
&\text{ since the minimum of $Y^3 - 3Y$ is $-2$ on $[0,1]$,}\\
s_{10} - s_8 &= 1 + \frac{1}{6} XY^3Z- \frac{1}{6} XYZ - YZ\\
&= 1 + \frac Z6 ( XY^3 - XY - 6Y)>1+\frac 16(-6)=0,\\
&\text{since $ XY^3 - XY - 6Y$
has minimum $-6$ on the unit square when $Y=1$.}
\end{align*}

\noindent For $\sigma_{ps}$,
\begin{align*}
s_1 - s_2 &= 1 +\frac{1}{18}XY -\frac{1}{18}XYZ^6 -X\\
&= 1 + \frac{1}{18}XY(1- Z^6) -X >1+0-1=0,\\
s_1 - s_3 &= 1 - \frac{1}{3}XY- \frac{2}{3}X >1-\frac 13-\frac23=0,\\
s_1 - s_8 &= 1 +\frac{1}{18}XY - \frac{1}{18}XYZ^6 - \frac{2}{3}X - \frac{1}{3}Y \\
&= 1 + \frac{1}{18}XY(1 - Z^6) - \frac{2}{3}X - \frac{1}{3}Y >1+0-\frac23 -\frac 13= 0,\\
s_{10} - s_2 &= 1 + \frac{1}{18}XY - \frac{1}{18}XYZ^6 - \frac{1}{3}X - \frac{2}{3}Y\\
&= 1 + \frac{1}{18}XY(1 - Z^6) - \frac{1}{3}X - \frac{2}{3}Y >1+0-\frac 13-\frac 23= 0,\\ 
s_{10} - s_3 &= 1 - \frac{1}{3} XY- \frac{2}{3}Y >1-\frac 13- \frac 23= 0,\\ 
s_{10} - s_8 &= 1 + \frac{1}{18}XY - \frac{1}{18} XYZ - Y\\ 
&= 1 + \frac{1}{18}XY(1 - Z) - Y >1+0-1=0.
\end{align*}
\end{proof}

\begin{proof}[Proof of Proposition \ref{prop:greedycat}]
We begin by showing that $s_1> s_i$ for $i = 2, \dots, 9$, using $0<X,Y,Z<1$.
One need only check that
\begin{align*}
s_1 - s_2 &= 1 - X >0,\\
s_1 - s_3 &= 1-\frac{1}{6}XY^3-\frac{1}{6}XY-\frac{2}{3}X > 0,\\
s_1 - s_4 &= 1-\frac{1}{18}XY^3 Z^6 - \frac{1}{9}XY^3 -\frac{1}{6}XY-\frac{2}{3}X >0,\\
s_1 - s_8 &= 1 -\frac{1}{6}XY^3+\frac{1}{6}XY-\frac{2}{3}X-\frac{1}{3}Y\\
&= 1 + \frac{1}{6}XY(1-Y^2)-\frac{2}{3}X-\frac{1}{3}Y  > 0,\\
s_1 - s_9 &= 1 -\frac{1}{18}XY^3 Z^6-\frac{1}{9}XY^3+\frac{1}{6}XY-\frac{2}{3}X-\frac{1}{3}Y \\
&= 1 +\frac{1}{18}XY(3-Y^2 Z^6-2Y^2)-\frac{2}{3}X-\frac{1}{3}Y>0\\
&\text{\ \ \ since $3-Y^2 Z^6-2Y^2 > 0$.}
\end{align*}

Suppose now that $s_{10} \ge s_1$.  Then since $s_1$ is larger than all the remaining 
split probabilities by the above calculations, the true non-trivial splits on the species tree have the highest probability, and greedy consensus for
gene tree splits  is consistent.

Now assume instead that $s_1>s_{10}$, so $s_1$ is the strict maximum of the split probabilities.  
Under the greedy consensus
algorithm, splits incompatible with $Sp(ab)$ are discarded and only the splits $s_8$, 
$s_9$, and $s_{10}$ remain as candidate splits for acceptance by the algorithm.

Noting that 
\begin{align*}
s_{10} - s_9 &= 1 + \frac{1}{6}XY^3 - \frac{1}{6}XY - Y\\
&= (1-Y)\left(1-\frac{1}{6} XY(1+Y) \right),
\end{align*}
and that $(1-\frac{1}{6} XY(1+Y) \big) > \frac{2}{3}$,
it follows that $s_{10} > s_9$.

Consider now 
$$s_{10} - s_8 = 1 + \frac{1}{18}XY^3Z^6 + \frac{1}{9}XY^3 - \frac{1}{6}XY - Y = F(X,Y,Z).$$
If $F(X,Y,Z)>0$, then greedy consensus will return the correct species tree.
If $F(X,Y,Z)<0$, it will return the tree $((a,b),e,(c,d))$.
\end{proof}

\section{Non-linear split invariants for 5-taxon trees}\label{app:grobner}

While non-linear split invariants exist for species trees with 5 or
more taxa, using {\tt Singular} we were only able to compute them for
5-taxon trees. We record results here, using the enumeration given in
Example \ref{ex:5splits}.

For the caterpillar species tree $((((a,b),c),d),e)$, the ideal of
invariants for non-trivial split probabilities is generated by the
trivial invariant and the following eight polynomials:

$$s_2-s_5,\qquad \qquad
s_3-s_6,$$
$$s_4-s_7,\qquad\qquad
s_6-s_7-s_8+s_9,$$
\begin{multline*}s_1^2+s_1s_5-2s_5^2+14s_1s_7+4s_5s_7+24s_7^2+9s_1s_8-9s_5s_8+18s_7s_8-5s_1s_9-13s_5s_9-30s_7s_9\\-18s_8s_9+6s_9^2-s_1s_{10}+s_5s_{10}+6s_7s_{10}-6s_9s_{10},
\end{multline*}
\begin{multline*}3780s_5^2s_7-18163s_1s_7^2-4046s_5s_7^2-36060s_7^3-1512s_1s_5s_8+1512s_5^2s_8-27516s_1s_7s_8\\
+9444s_5s_7s_8-54363s_7^2s_8-10296s_1s_8^2+10296s_5s_8^2-21132s_7s_8^2-1008s_1s_5s_9-2772s_5^2s_9\\
+14362s_1s_7s_9+20768s_5s_7s_9+65529s_7^2s_9+10356s_1s_8s_9+7716s_5s_8s_9+75246s_7s_8s_9\\
+21132s_8^2s_9-3063s_1s_9^2-9858s_5s_9^2-35226s_7s_9^2-20883s_8s_9^2+5757s_9^3+3620s_1s_7s_{10}\\
-1340s_5s_7s_{10}-8397s_7^2s_{10}+2424s_1s_8s_{10}-2424s_5s_8s_{10}-6744s_7s_8s_{10}-2004s_1s_9s_{10}\\
-276s_5s_9s_{10}+12978s_7s_9s_{10}+6744s_8s_9s_{10}-4581s_9^2s_{10}+420s_7s_{10}^2-420s_9s_{10}^2,
\end{multline*}
\begin{multline*}63s_1s_5s_7+126s_5^2s_7-506s_1s_7^2-28s_5s_7^2-984s_7^3-771s_1s_7s_8+321s_5s_7s_8-1518s_7^2s_8\\
-288s_1s_8^2+288s_5s_8^2-603s_7s_8^2-63s_1s_5s_9-126s_5^2s_9+401s_1s_7s_9+505s_5s_7s_9+1818s_7^2s_9\\
+291s_1s_8s_9+159s_5s_8s_9+2118s_7s_8s_9+603s_8^2s_9-87s_1s_9^2-285s_5s_9^2-999s_7s_9^2-600s_8s_9^2\\
+165s_9^3+139s_1s_7s_{10}-25s_5s_7s_{10}-234s_7^2s_{10}+96s_1s_8s_{10}-96s_5s_8s_{10}-186s_7s_8s_{10}\\
-75s_1s_9s_{10}-39s_5s_9s_{10}+378s_7s_9s_{10}+186s_8s_9s_{10}-144s_9^2s_{10}+21s_7s_{10}^2-21s_9s_{10}^2,
\end{multline*}
\begin{multline*}
19845s_1s_5^2+39690s_5^3+209186s_1s_7^2-49028s_5s_7^2+467400s_7^3+16254s_1s_5s_8+20601s_5^2s_8\\
+310092s_1s_7s_8-257748s_5s_7s_8+670146s_7^2s_8+112797s_1s_8^2-183672s_5s_8^2+212904s_7s_8^2\\
-25515s_8^3+25326s_1s_5s_9+102249s_5^2s_9-158744s_1s_7s_9+20084s_5s_7s_9-720918s_7^2s_9\\
-112962s_1s_8s_9+100308s_5s_8s_9-785172s_7s_8s_9-198729s_8^2s_9+33891s_1s_9^2+55176s_5s_9^2\\
+381852s_7s_9^2+219921s_8s_9^2-63129s_9^3+11970s_1s_5s_{10}+4095s_5^2s_{10}-52900s_1s_7s_{10}\\
-93860s_5s_7s_{10}-127746s_7^2s_{10}-35898s_1s_8s_{10}-49152s_5s_8s_{10}-236352s_7s_8s_{10}\\
-112455s_8^2s_{10}+32538s_1s_9s_{10}+89652s_5s_9s_{10}+62244s_7s_9s_{10}+71922s_8s_9s_{10}\\
+13527s_9^2s_{10}+2205s_1s_{10}^2-7560s_5s_{10}^2-31080s_7s_{10}^2-21105s_8s_{10}^2+20055s_9s_{10}^2\\
-2205s_{10}^3
\end{multline*}
The four linear invariants here are all given by theorems in the text,
but we have no theoretical explanation for the form of the quadratic
and 3 cubics.

\medskip

For the species tree $(((a,b),c),(d,e))$, the ideal of invariants for non-trivial split probabilities is generated by the trivial invariant and the following six polynomials:
$$s_2-s_5,\qquad\qquad s_3-s_6,\qquad\qquad s_3-s_4$$
$$s_6-s_7,\qquad\qquad s_8-s_9$$
\begin{multline*}
s_1^2+s_1s_5-2s_5^2+14s_1s_7+4s_5s_7+24s_7^2+4s_1s_9-22s_5s_9-12s_7s_9-12s_9^2-s_1s_{10}+s_5s_{10}\\
+6s_7s_{10}-6s_9s_{10}
\end{multline*}
The five linear polynomials all arise from clade-induced constructions
in the text, but the quadratic has not been explained.

\medskip

For the species tree $(((a,b),(d,e)),c)$, the ideal of invariants for
non-trivial split probabilities is generated by the trivial invariant
and the following six polynomials:
$$s_2-s_5,\qquad\qquad
s_3-s_6,\qquad\qquad
s_4-s_7$$
$$s_6-s_7,\qquad\qquad
s_8-s_9,\qquad\qquad s_1+2s_5-2s_9-s_{10}$$
Note all are linear, with
the first five given by the general clade-induced construction, and
the last being explained by Proposition \ref{prop:special5}.
 
 \section{Additional Proofs}\label{app:proofs}
 
The proofs we give of Propositions \ref{prop:acbc}, \ref{prop:abbc}, and \ref{prop:aabbc} depend on a careful analysis of probabilities under the coalescent. That of Proposition \ref{prop:acbc} is the simplest, and serves as a model for the others.

\subsection{Proof of Proposition \ref{prop:acbc}} \ 

The proof of \ref{prop:acbc} depends on several lemmas.

We begin with a definition. Consider a non-binary rooted species tree
$((x_1,x_2,\dots x_k)\tc L,y)$ formed by attaching a single outgroup
taxon $y$ to a claw tree with $k$ taxa $x_i$, with edge length $L>0$.
Under the multispecies coalescent model we will be interested in the
case where the gene lineages, one for each $x_i$, have coalesced
$\ell$ times, from $k$ to $k-\ell$ lineages, by the time they reach
the root of the tree, and then further coalescences occur with the $y$
lineage in the root population, until a single tree is formed. For
$\mathcal A\subset\mathcal X= \{x_1,x_2,\dots x_k,y\}$. We denote the
probability that a resulting gene tree displays a split $Sp(\mathcal
A_g)$ as
$$p(\mathcal A\mid k,\ell).$$
Note that this probability does not depend on branch lengths in the
species tree, since $L>0$ and we have conditioned on
$\ell$. Furthermore, since the $x_i$ lineages are exchangeable under
the coalescent model on this tree, $p(\mathcal A\mid k,\ell)$ actually
depends on $\mathcal A$ only through the number of $x_i\in\mathcal A$
and whether $y\in A$, but not on the particular $x_i\in \mathcal A$.

By an $m$-\emph{split}, we mean a split of taxa where one block of the partition has size $m$.
We now give recursions and base cases for 
 the probability of various 
$2$-splits for the above species tree.
\begin{lemma}\label{lem:recur} \ 
\begin{enumerate}
\item\label{case:acbck0} $p(x_1x_2\mid k,0)=p(x_1y\mid k, 0)$ for $k\ge 2$,
\item\label{case:3l} $p(x_1x_2\mid 3,\ell)=p(x_1y\mid 3, \ell)=\frac 13$ for $\ell=0,1,2$,
\item\label{case:bck0} $p(x_1y\mid k, 0)=\frac {1}{\binom{k+1}2} + \frac{\binom{k-1}2}{\binom{k+1}2}p(x_1y\mid k-1, 0)$ for $k\ge3$,
\item \label{case:ackl}$p(x_1x_2\mid k, \ell)=\frac {1}{\binom{k}2} + \frac{\binom{k-2}2}{\binom{k}2}p(x_1x_2\mid k-1, \ell-1)$, for $k\ge 4$, $k> \ell\ge 1$, 
\item\label{case:bckl} $p(x_1y\mid k, \ell)=\frac{\binom{k-1}2}{\binom{k}2} p(x_1y\mid k-1, \ell-1)$ for $k\ge 3$, $k>\ell\ge1$.
\end{enumerate}
\end{lemma}
\begin{proof} These all follow directly from properties of the coalescent model. We give reasoning for several, leaving the rest to the reader.

  For claim \eqref{case:acbck0}, observe no coalescent events occur
  below the root of the tree, so exchangeability of lineages at the
  root implies the statement.

  For claim \eqref{case:ackl}, note that for the split $Sp(X_1X_2)$ to
  form, the first coalescent event must either be between the $x_1$
  and $x_2$ lineages, which occurs with probability $1/\binom k2$, or
  be between $x_i$ lineages with $i\ne 1,2$, which occurs with
  probability $\binom{k-2}2/\binom k2$, with the split forming
  subsequently.
\end{proof}

We next establish some probability bounds.

\begin{lemma} \label{lemESA:bounds}  For $k\ge 4$, $k > \ell \ge 0$,
$p(x_1y \mid k,\ell) < \frac 1k$.
\end{lemma}
\begin{proof}   Lemma \ref{lem:recur} \eqref{case:bck0} and \eqref{case:3l} imply
$p(x_1y\mid 4,0)= \frac 15 < \frac 14$.
For $k>4$, $\ell = 0$,  Lemma \ref{lem:recur} \eqref{case:bck0}  and an inductive hypothesis then show
$$p(x_1y\mid k,0)<\frac 1{\binom{k+1}2} +\frac {\binom{k-1}2}{\binom{k+1}2}\frac 1{k-1}= \frac  1{k+1} <\frac 1{k}.$$
For $\ell \ge 1$, first consider the case that $k-\ell=1$, $2$, or $3$. The using 
Lemma  \ref{lem:recur}  \eqref{case:bckl}   repeatedly and 
Lemma  \ref{lem:recur}  \eqref{case:3l} shows
\begin{align*}p(x_1y\mid k, \ell)&=\frac{\binom{k-1}2}{\binom{k}2} \frac{\binom{k-2}2}{\binom{k-1}2}\cdots \frac{\binom{3}2}{\binom{4}2}p(x_1y \mid 3, \ell-k+3)\\
&=\frac {6}{k(k-1)} \cdot \frac 13=\frac 2{k(k-1)}< \frac 1k.
\end{align*}
If instead $k-\ell\ge 4$,  Lemma \ref{lem:recur}  \eqref{case:bckl} and what has already been established imply
\begin{align*}
p(x_1y\mid k, \ell)&=\frac{\binom{k-1}2}{\binom{k}2} \frac{\binom{k-2}2}{\binom{k-1}2}\cdots \frac{\binom{k-\ell}2}{\binom{k-\ell+1}2}p(x_1y \mid k-\ell, 0)\\
&\le\frac {(k-\ell)(k-\ell-1)}{k(k-1)}\frac 1{k-\ell}<\frac 1k.
\end{align*}
\end{proof}

Next, we obtain a key inequality.  
\begin{lemma}\label{lem:ineq} For $k=3$, $\ell=0,1,2$ and for $k > 3$, $\ell=0$,
$$p(x_1x_2\mid k,\ell)-p(x_1y\mid k,\ell)=0.$$
For $k\ge 4$ and $k>\ell\ge1$,
$$p(x_1x_2\mid k,\ell)-p(x_1y\mid k,\ell)>0.$$
\end{lemma}
\begin{proof}
For $k=3$, $\ell=0,1,2$ and for $k>3$, $\ell=0$, the claimed equalities follow from 
Lemma \ref{lem:recur} \eqref{case:3l} and \eqref{case:acbck0}, respectively.

\smallskip

For the inequality when $k\ge 4$, $k>\ell\ge 1$, by Lemma \ref{lem:recur}  \eqref{case:ackl} and \eqref{case:bckl}, 
\begin{align}
p(x_1&x_2\mid k,\ell)-p(x_1y\mid k,\ell)\notag \\ &=\frac 1{\binom k2}\Big ( 1+\binom{k-2}2p(x_1x_2\mid k-1,\ell-1)-\binom{k-1}2p(x_1y\mid k-1,\ell-1)\Big )\notag\\
&=\frac 1{\binom k2}\Big ( 1+\binom{k-2}2\big( \, p(x_1x_2\mid k-1,\ell-1)-p(x_1y\mid k-1,\ell-1)\, \big)\label{eq:reduc}\\
&\qquad \qquad \qquad \qquad \qquad \qquad\ \ -(k-2)p(x_1y\mid k-1,\ell-1) \Big ).\notag
\end{align}
 
Using Lemma \ref{lem:recur}  \eqref{case:3l} in equation \eqref{eq:reduc} shows
$$p(x_1x_2\mid 4,\ell)-p(x_1y\mid 4,\ell) = \frac{1}{18} > 0$$
for $\ell=1,2,3$, establishing the $k=4$ case of the inequality.

For $k > 4$, 
by Lemma \ref{lem:recur} \eqref{case:acbck0} 
and Lemma \ref{lemESA:bounds}  
equation \eqref{eq:reduc} yields
$$
p(x_1x_2\mid k,1)-p(x_1y\mid k,1)>\frac 1{\binom k2}\Big ( 1+\binom{k-2}2 0-(k-2)\frac 1{k-1} \Big )>0.
$$
This shows the inequality holds for $\ell=1$, and provides base cases for an inductive proof for $\ell\ge1$.

Finally, equation \eqref{eq:reduc}, an inductive hypothesis, and  Lemma 
\ref{lemESA:bounds} show that for $\ell\ge 2$
$$
p(x_1x_2\mid k,\ell)-p(x_1y\mid k,\ell)>\frac 1{\binom k2}\Big ( 1+\binom{k-2}2 0 -(k-2)\frac 1{k-1} \Big) >0.
$$
\end{proof}

\begin{proof}[Proof of Proposition \ref{prop:acbc}]
  That the equality holds for species tree $(T,(a,b))$ is an instance
  of Theorem \ref{thm:cladesplitinv}.  It is enough to establish the
  inequality for $((T,a),b)$, since that for $((T,b),a)$ will follow
  by interchanging taxon names.

  On the species tree $((T,a),b)$, let $v$ denote the MRCA of $T$ and
  $a$.  Observe that for the splits $Sp(AC)$ or $Sp(BC)$ to form, it
  is necessary that the $c$ lineage not coalesce with any other below
  $v$.  In any such realization of the coalescent process below $v$,
  lineages from taxa on $T$ will have coalesced to $k-1$ lineages by
  $v$, where $k\ge 3$. There the lineage from $a$ enters, and $\ell$
  coalescent events, $k>\ell\ge 0$ occur on the edge immediately
  ancestral to $v$.

  To establish the inequality, we consider it conditioned on a number
  of disjoint and exhaustive events: For each possible $k, \ell$, let
  $\mathcal C = \mathcal C (k, \ell)$ denote the event that $k-1$
  agglomerated lineages from $T$ reach $v$, one of which is the
  lineage from $c$ alone, and that $\ell$ coalescent events occur in the
  population immediately ancestral to $v$.  Fixing $\mathcal C =
  \mathcal C (k, \ell)$, with $y=b$, $x_1=c$, $x_2=a$ we have
\begin{align*} \PP(Sp(ac)\mid \mathcal C) &= p( x_1x_2 \mid k,\ell),\\
\PP(Sp(bc)\mid \mathcal C)&=p(x_1y \mid k, \ell).
\end{align*} 
Lemma \ref{lem:ineq} thus shows $\PP(Sp(ac\mid \mathcal C)-\PP(Sp(bc\mid \mathcal C))$ is positive 
for $k\ge 4$, $k>\ell\ge 1$, and zero for other relevant cases. 
Multiplying by the probabilities of each $\mathcal C = \mathcal C (k, \ell)$ and summing, we obtain the desired unconditioned expression
$\PP(Sp(ac))-\PP(Sp(bc))$.
Because $T$ has at least three taxa, there are some positive summands 
from  $k \ge 4$, $\ell \ge 1$, so the desired inequality holds.
\end{proof}

\subsection{Proof of Proposition \ref{prop:abbc}} \ 

While the proof of Proposition \ref{prop:abbc} follows the same line
of reasoning as that of Proposition \ref{prop:acbc}, there are further
technical details.  We first extend some of the results from the
previous section to splits of size 3. These will be applied in
arguments for the species tree $((T,(b_1,b_2)),a)$.
 
\begin{lemma} \label{lem:recurbba} \  
\begin{enumerate}
\item \label{case:xxy0} $p(\mathcal A\mid k,0)=p(\mathcal B\mid k,0)$ for $|\mathcal A|=|\mathcal B|,$

\item \label{case:xxxk0} $p(x_1x_2x_3\mid k,0)=\frac 3{\binom{k+1}2} p(x_1x_2\mid k-1,0)+\frac{\binom{k-2}2}{\binom{k+1}2}p(x_1x_2x_3\mid k-1, 0)$ for $k\ge4$,

\item  \label{case:xxy3l}$p(x_1x_2x_3\mid 3,\ell)=p(x_1x_2y\mid 3,\ell)=1$, for $\ell=0,1,2$,

\item \label{case:xxx4l} $p(x_1x_2x_3\mid 4,\ell)=\frac 1{2}p(x_1x_2\mid 3,\ell-1)$ for $\ell = 1, 2, 3$, 

\item \label{case:xxxkl} $p(x_1x_2x_3\mid k,\ell)=\frac 3{\binom k2}p(x_1x_2\mid k-1,\ell-1)+\frac{\binom{k-3}2}{\binom k2}p(x_1x_2x_3\mid k-1,\ell-1)$ for $k\ge 5$, $k>\ell\ge 1$, 

\item \label{case:xxykl} $p(x_1x_2y\mid k,\ell)=\frac 1{\binom k2}p(x_1y\mid k-1,\ell-1)+\frac{\binom{k-2}2}{\binom k2}p(x_1x_2y\mid k-1,\ell-1)$ for $k\ge 4$, $k>\ell\ge 1.$ 
\end{enumerate}
\end{lemma}

\begin{proof}
  For claim \eqref{case:xxy0}, it suffices to note that $k+1$ lineages
  enter the population above the root, with no coalescent events
  having occurred below, so the probabilities of any two $m$-splits
  are the same by exchangeability of lineages under the coalescent
  model.

  For claim \eqref{case:xxxk0}, again $k+1$ lineages enter the root
  population, with no previous coalescence.  For $Sp(X_1X_2X_3)$ to
  form, the first coalescent event above the root must be between a
  pair of lineages chosen from $x_1, x_2, x_3$, or disjoint from them.
  It is between a pair chosen from them with probability
  $\frac{3}{\binom{k+1}{2}}$.  Then, for $Sp(X_1 X_2 X_3)$ to form,
  this pair's lineage must join with the remaining $x_i$ lineage.  By
  claim (1), this has probability $p(x_1 x_2\mid k-1,0)$.  Multiplying
  these probabilities, we obtain the first summand.  The first
  coalescent event not involving any of the $x_1,x_2,x_3$ lineages, and then the
  desired split forming with 1 less lineage present gives the second
  summand.

The remaining verifications are left to the reader.
\end{proof}

\begin{lemma}\label{lem:boundsbba} For $k\ge 4$, $k>\ell\ge 0$,
$$p(x_1y\mid k,\ell)+p(x_1x_2y\mid k,\ell)<\frac 1{k-2}.$$ 
\end{lemma}
\begin{proof} We first show the inequality for $\ell=0$, by induction on $k$.
From Lemmas \ref{lem:recur} and \ref{lem:recurbba},
$p(x_1y\mid 4,0)+p(x_1x_2y\mid 4,0)=\frac 25,$
establishing the base case of $k=4$.  
 
If $k\ge 5$,  an inductive hypothesis, 
Lemma \ref{lem:recur} \eqref{case:bck0}, Lemma \ref{lem:recurbba} \eqref{case:xxy0} and \eqref{case:xxxk0}, and
Lemma  \ref{lemESA:bounds} 
yield
\begin{align*}p(x_1&y\mid k,0)+p(x_1x_2y\mid k,0)\\
&=\frac 1{\binom{k+1}2}\left ( 1+ (k+1)p(x_1y\mid k-1,0) +\binom{k-2}2\left(p(x_1y\mid k-1,0)+p(x_1x_2y\mid k-1,0)\right) \right )\\
&<  \frac 1{\binom{k+1}2}\left (1+\frac {k+1}{k-1}+ \binom{k-2}2\frac 1{k-3} \right )
=\frac{k^2+k+2}{(k+1)k(k-1)}  <\frac 1{k-2}.
\end{align*}

Next observe that for $\ell=1,2,3$, Lemmas \ref{lem:recur} and \ref{lem:recurbba} imply
$$
p(x_1y\mid 4, \ell)+p(x_1x_2y\mid 4,\ell)=\frac 7{18}<\frac 1{4-2}.
$$
With the $k=4$, $\ell=1,2,3$ cases and the $k\ge 4$, $\ell=0$ cases already established,  we now proceed by induction on $\ell$.
For
 $k\ge 5$, $k>\ell\ge 1$ by Lemma \ref{lem:recur} \eqref{case:bckl}, 
 Lemma \ref{lem:recurbba} \eqref{case:xxykl}, Lemma \ref{lemESA:bounds}, and an inductive hypothesis, 
\begin{align*}p(x_1&y\mid k,\ell)+p(x_1x_2y\mid k,\ell)\\
&=\frac {\binom{k-1}2 +1}{\binom k2}p(x_1y\mid k-1,\ell-1) +\frac{\binom{k-2}2}{\binom k2}p(x_1x_2y\mid k-1,\ell-1)\\
&=\frac {k-1}{\binom k2}p(x_1y\mid k-1,\ell-1) +\frac{\binom{k-2}2}{\binom k2}\left (p(x_1y\mid k-1, \ell-1)+p(x_1x_2y\mid k-1,\ell-1)\right )\\
&<\frac {k-1}{\binom k2}\frac 1{k-1} +\frac{\binom{k-2}2}{\binom k2}\frac 1{k-3}=\frac 1{k-1}<\frac 1{k-2}.\end{align*}
\end{proof}

\begin{lemma} \label{lem:ineqbba} Let 
$$P(k,\ell)=p(x_1x_2\mid k,\ell)+p(x_1x_2x_3\mid k,\ell)-p(x_1y\mid k,\ell)-p(x_1x_2y\mid k,\ell).$$
Then for  $k=4$, $\ell=0,1,2,3$, and for $k\ge 5$, $\ell=0$, $P(k,\ell)=0.$
For $k\ge 5$, $k>\ell\ge1$,
$P(k,\ell)>0.$
\end{lemma}

\begin{proof}
Note that for $k=4$, the event $Sp(x_1x_2)$ is the same as $Sp(x_3x_4y)$, and $Sp(x_1x_2x_3)$ is the same as $Sp(x_4y)$, so
using exchangeability of the $x_i$ lineages we have
\begin{align*}
p(x_1x_2\mid 4,\ell)&=p(x_1x_2y\mid 4,\ell),\\
p(x_1x_2x_3\mid k,\ell)&=p(x_1y\mid k,\ell).
\end{align*}
Thus $P(4,\ell)=0$ for $\ell = 0, 1, 2, 3$.  
For $k\ge 5$,  Lemma \ref{lem:recurbba} \eqref{case:xxy0} implies $P(k,0)=0$.

For $k\ge 5$, $\ell \ge 1$, by Lemmas \ref{lem:recur} \eqref{case:ackl}, \eqref{case:bckl} and 
\ref{lem:recurbba} \eqref{case:xxxkl}, \eqref{case:xxykl}  we find
\begin{align*}
P(k,\ell)&=\frac 1{\binom k2}\bigg [ 1+(k-2) \left(p(x_1x_2\mid k-1,\ell-1) -p(x_1y\mid k-1,\ell-1)\right)\\
& -(k-2) \left (p(x_1y\mid k-1,\ell-1)+p(x_1x_2y\mid k-1,\ell-1) \right )\\
&+2p(x_1x_2\mid k-1,\ell-1)
+p(x_1x_2y\mid k-1,\ell-1)+\binom {k-3}2 P(k-1,\ell-1)\bigg ].
\end{align*}
Using Lemmas \ref{lem:ineq} and \ref{lem:boundsbba}, the non-negativity of probabilities, and an inductive hypothesis
that $P(k-1,\ell-1)\ge 0$, it follows that
$$
P(k,\ell)>\frac 1{\binom k2}\Big ( 1+(k-2)\cdot 0 -(k-2) \frac 1{k-2}+2\cdot0
+0+\binom {k-3}2 0 \Big )=0.$$
\end{proof}

\begin{lemma} \label{lem:bba} Consider a species tree with topology
  $((T,(b_1,b_2)),a)$, where $T$ is a subtree on at least three taxa,
  one of which is $c$. Suppose the edge above $(T,(b_1,b_2))$ has
  positive length.  Then under the multispecies coalescent model,
$$\PP(Sp(ac))+\PP(Sp(ab_2c))-\PP(Sp(b_1c))- \PP(Sp(b_1b_2c))<0.$$
\end{lemma}

\begin{proof}
Let $v$ denote the MRCA  on the species tree of the taxa on $T$ and the $b_i$.

To establish 
the claimed inequality, it is enough to show it holds when conditioned on whether
$b_1$ and $b_2$ lineages have coalesced before reaching $v$ or not. If they have coalesced before $v$ 
to form a single lineage, then the events $Sp(ab_2c)$ and $Sp(b_1c)$ have probability zero. 
Thus using $b$ for $b_1b_2$ we wish to show
$$\PP(Sp(ac))-\PP(Sp(bc))<0.$$
This follows immediately from Proposition \ref{prop:acbc}.

We henceforth condition on the two $b_i$ lineages being distinct at
$v$.  Noticing that all four probabilities in the expression of
interest are 0 if the $c$ lineage coalesces with any lineage below $v$, we further
condition on the $c$ lineage being distinct at $v$, where there are thus
$k\ge 4$ lineages entering the population above $v$, and $\ell$
coalescent events occuring between $v$ and the root.

Then, with $\mathcal C=\mathcal C(k,\ell)$ denoting the events we condition on, 
\begin{align*}
\PP(Sp(ac)\mid \mathcal C)&=p(x_1y\mid k,\ell),\\
\PP(Sp(ab_2c)\mid \mathcal C)&=p(x_1x_2y\mid k, \ell),\\
\PP(Sp(b_1c)\mid \mathcal C)&=p(x_1x_2\mid k,\ell),\\
\PP(Sp(b_1b_2c) \mid \mathcal C)&=p(x_1x_2x_3\mid k,\ell).
\end{align*}
From Lemma \ref{lem:ineqbba} we find conditioned on $\mathcal C$ that the expression 
is strictly negative for 
$k\ge 5$, $k > \ell\ge 1$, 
and zero for $k\ge 5$, $\ell=0$ and $k=4$, $\ell=0,1,2,3$.  
Thus weighting the conditioned expressions by the probabilities of the events 
$\mathcal C$ and summing, 
we see the full expression is 
negative,
as long as  $k \ge 5$ and  $\ell\ge 1$  is possible.
Since $T$ has at least $3$ taxa, this only requires 
that the edge above $v$ has positive length.
\end{proof}

\medskip

To handle the species tree $((T,a),(b_1,b_2))$ we proceed analogously,
but consider a rooted species tree $((x_1,x_2,\dots x_k)\tc
L,y_1,y_2)$ formed by attaching a trifurcating root to two outgroups
$y_1,y_2$ and a claw tree with $k$ taxa, with a positive edge length
$L$. We will be interested in the case where the gene lineages, one
for each $x_i$, have coalesced $\ell$ times, from $k$ to $k-\ell$
lineages, by the time they reach the root of the tree, and then
further coalescence occurs in the root population until a single tree
is formed. With $\mathcal X=\{x_1,\dots x_k,y_1,y_2\}$ and $\mathcal
A\subset \mathcal X$, let
$$r(\mathcal A\mid k,\ell)=\PP(Sp(\mathcal A )\mid k,\ell)$$
for this species tree. By exchangeability of lineages in the
coalescent model, $r(\mathcal A\mid k,\ell)$ depends on $\mathcal A$
only up to the number of $x_i$ and the number of $y_i$ it contains.

The reader who has followed previous arguments should be able to verify the following.

\begin{lemma}\label{lem:recurabb} \ 
\begin{enumerate}
\item\label{case:abb0} $r(\mathcal A \mid k,0)=r(\mathcal B\mid k, 0)$ for $|\mathcal A|=|\mathcal B|$,
\item\label{case:abbxx3l} $r(x_1 x_2 \mid 3, 0) = \frac 15$,   $r(x_1x_2\mid 3,\ell)=\frac 13$ for $\ell=1,2$,
\item\label{case:abbxyk} $r(x_1x_2\mid k, 0)=\frac {1}{\binom{k+2}2} + \frac{\binom{k}2}{\binom{k+2}2}r(x_1x_2\mid k-1, 0)$ for $k\ge3$,

\item \label{case:abbxxkl}$r(x_1x_2\mid k, \ell)=\frac {1}{\binom{k}2} + \frac{\binom{k-2}2}{\binom{k}2}r(x_1x_2\mid k-1, \ell-1)$, for $k\ge 4$, $k> \ell\ge 1$, 

\item\label{case:abbxy2} $r(x_1y_1 \mid 2,0)=\frac 13$, $r(x_1y_1 \mid 2,1)=0$,

\item\label{case:abbxykl} $r(x_1y_1\mid k, \ell)=\frac{\binom{k-1}2}{\binom{k}2} r(x_1y_1\mid k-1, \ell-1)$ for $k\ge 3$, $k>\ell\ge1$.

\item $r(x_1x_2x_3 \mid 4,\ell)= \frac{1}{2} r(x_1x_2\mid 3, \ell-1)$ for $\ell = 1, 2, 3$,  

\item $r(x_1x_2x_3 \mid k,\ell)=\frac 3{\binom k2} r(x_1x_2\mid k-1, \ell-1)+\frac {\binom{k-3}2}{\binom k2} r(x_1x_2x_3\mid k-1, \ell-1)$ for $k\ge 5$, $k>\ell\ge 1$,  

\item \label{case:abbxxy23} 
$r(x_1x_2y_1\mid 2,0)=r(x_1x_2y_1\mid 2,1) =1$, $r(x_1x_2y_1\mid 3,1)=\frac19$, $r(x_1x_2y_1\mid 3,2)=0$,

\item \label{case:abbxxyk0}$r(x_1x_2y_1\mid k, 0)=\frac 3{\binom{k+2}2}r(x_1x_2\mid k-1,0)+\frac {\binom{k-1}2}{\binom{k+2}2} r(x_1x_2y_1\mid k-1, 0)$ for $k\ge 3$,

\item \label{case:abbxxykl}$r(x_1x_2y_1\mid k, \ell)=\frac 1{\binom k2}r(x_1y_1\mid k-1,\ell-1)+\frac {\binom{k-2}2}{\binom k2}r(x_1x_2y_1\mid k-1, \ell-1)$ for $k\ge 4$, $k>\ell\ge1$,

\item 
$r(x_1 y_1 y_2 \mid 2, 0) =  r(x_1y_1y_2\mid 2, 1)=1$, 

\item \label{case:abbxyykl}$r(x_1y_1y_2\mid k, \ell)=\frac{\binom{k-1}2}{\binom k 2}r(x_1y_1y_2\mid k-1, \ell-1)$ for $k\ge 3$, $k>\ell\ge 1$.
\end{enumerate}
\end{lemma}

\begin{lemma} \label{lem:boundsabb} \ 
\begin{enumerate}
\item\label{caseb:abbxyk0} $r(x_1y_1 \mid k,0)\le \frac 1{k+2}$ for $k\ge 3$,
\item\label{caseb:abbxyykESA} $r(x_1y_1 \mid k,\ell)+r(x_1y_1y_2 \mid k,\ell)< \frac 1{k-1}$ if $k \ge 3$ and $\ell = 0$,
or if $k\ge 4$ and $k> \ell\ge 1$.
\end{enumerate}
\end{lemma}
\begin{proof} 
For claim \eqref{caseb:abbxyk0} first note that  Lemma \ref{lem:recurabb} \eqref{case:abb0} and \eqref{case:abbxx3l} establish the $k=3$ case.
Then using  Lemma \ref{lem:recurabb} \eqref{case:abbxyk} one sees inductively that for $k>3$,
\begin{align*}r(x_1y_1\mid k, 0)&\le\frac {1}{\binom{k+2}2} + \frac{\binom{k}2}{\binom{k+2}2}\frac 1{k+1}
=\frac{k^2+k+2}{(k+2)(k+1)^2}\\
&\le\frac{(k+1)^2}{(k+2)(k+1)^2}=\frac 1{k+2}.
\end{align*}

For claim \eqref{caseb:abbxyykESA}  when $\ell = 0$, note that by 
Lemma \ref{lem:recurabb} \eqref{case:abb0}, \eqref{case:abbxx3l}, \eqref{case:abbxy2}, \eqref{case:abbxxy23}, and \eqref{case:abbxxyk0}, 
$$r(x_1y_1 \mid 3,0)+r(x_1y_1y_2\mid 3, 0)=\frac 15 +\frac 3{10}\cdot\frac 13 +\frac 1{10}\cdot 1=\frac 2{5}<\frac 12,$$
so the base case of $k=3$ holds.
Then for $k> 3$, using Lemma \ref{lem:recurabb} \eqref{case:abb0}, \eqref{case:abbxyk}, and \eqref{case:abbxxyk0}  we have
\begin{align*}
r(&x_1y_1\mid k, 0)+r(x_1y_1y_2\mid k, 0)\\
&=\frac 1{\binom{k+2}2} +\frac{\binom k2+3}{\binom{k+2}2} r(x_1y_1\mid k-1, 0)
+\frac {\binom{k-1}2}{\binom{k+2}2}r(x_1y_1y_2\mid k-1, 0)\\
&=\frac 1{\binom{k+2}2} +\frac{k+2}{\binom{k+2}2} r(x_1y_1\mid k-1, 0)
 +\frac {\binom{k-1}2}{\binom{k+2}2} \bigg( r(x_1y_1\mid k-1, 0)+
r(x_1y_1y_2\mid k-1, 0)\bigg).
\end{align*}
Using an inductive hypothesis and claim \eqref{caseb:abbxyk0} of this proposition yields
\begin{align*}
r(x_1y_1\mid k, 0)+r(x_1&y_1y_2\mid k, 0)\\
&<\frac 2{(k+2)(k+1)} +\frac{2}{k+1} \cdot \frac 1{k+1}
+\frac {(k-1)(k-2)}{(k+2)(k+1)} \cdot \frac 1{k-2}\\
&=\frac {1}{k+2} +\frac{2}{(k+1)^2} <\frac 1{k+1} +\frac{2}{(k+1)^2}<\frac 1{k-1}.\\
\end{align*}

Assume now $k \ge 4$ and $k > \ell \ge 1$, and consider first 
the case that $k-\ell=1$ or $2$. Applying Lemma \ref{lem:recurabb} \eqref{case:abbxykl} and \eqref{case:abbxyykl} repeatedly we have
$$r(x_1y_1\mid k, \ell)+r(x_1y_1y_2\mid k, \ell)=\frac 3{\binom k2}\big (r(x_1y_1\mid 3, \ell-k+3)+r(x_1y_1y_2\mid 3, \ell-k+3)\big ).$$
From Lemma \ref{lem:recurabb} 
\begin{align*}r(x_1y_1\mid 3,1)+r(x_1y_1y_2\mid 3,1)&=\frac 19+ \frac 13=\frac 49,\\
r(x_1y_1\mid 3,2)+r(x_1y_1y_2\mid 3,2)&=0+\frac 13=\frac 13,
\end{align*}
so for $k\ge 4$,
$$
r(x_1y_1\mid k, \ell)+r(x_1y_1y_2\mid k, \ell)<\frac 6{k(k-1)} \cdot \frac 49<  
\frac 1{k-1}.$$ 
If $k-\ell\ge 3$, then applying Lemma \ref{lem:recurabb} \eqref{case:abbxykl} and \eqref{case:abbxyykl} repeatedly gives
$$
r(x_1y_1\mid k, \ell)+r(x_1y_1y_2\mid k, \ell)=\frac {\binom{k-\ell}2}{\binom k2}\left (r(x_1y_1\mid k-\ell, 0)+r(x_1y_1y_2\mid k-\ell,0)\right ).$$
Using what we proved above, this shows
$$r(x_1y_1\mid k, \ell)+r(x_1y_1y_2\mid k, \ell)<\frac {\binom{k-\ell}2}{\binom k2}\cdot \frac 1{k-\ell-1}
<\frac 1{k-1}.$$
\end{proof}

\begin{lemma}\label{lem:ineqabb} 
Let 
$$R(k, \ell) = r(x_1x_2\mid k,\ell)+r(x_1x_2y_1\mid k,\ell)-r(x_1y_1\mid k,\ell)-r(x_1y_1y_2\mid k,\ell).$$
Then for $k=3$, $\ell=0,1,2$, and for $k\ge 4$, $\ell=0$, $R(k,\ell)=0.$ For $k\ge 4$ and $k>\ell\ge1$, $R( k,\ell)>0.$
\end{lemma}

\begin{proof}
For $k=3$, the events $Sp(x_1x_2)$ and $Sp(x_3y_1y_2)$ are the same, as are $Sp(x_1x_2y_1)$ and $Sp(x_3y_2)$, so using exchangability of the $x_i$ and of the $y_i$ lineages
\begin{align*}
r(x_1x_2\mid k, \ell)&=r(x_1y_1y_2\mid k, \ell),\\
r(x_1x_2y_1\mid k, \ell)&=r( x_1y_1 \mid k, \ell).
\end{align*}
Thus $R(3,\ell)=0$ for $\ell=0,1,2$. For $k\ge 3$, Lemma \ref{lem:recurabb} \eqref{case:abb0} implies $R(k,0)=0$.

Now consider $k\ge 4$, $k>\ell\ge 1$. By Lemma \ref{lem:recurabb} \eqref{case:abbxxkl}, \eqref{case:abbxykl}, \eqref{case:abbxxykl}, and
\eqref{case:abbxyykl} we find
\begin{multline*}R(k,l)=
\frac 1{\binom k2}\Big (  1 +r(x_1y_1 \mid k-1,\ell-1) +\binom{k-2}2 R(k-1,\ell-1) \\ 
-(k-2)\big (  r(x_1y_1\mid k-1, \ell-1) +r(x_1y_1y_2\mid k-1,\ell-1)    \big)
 \Big ).
\end{multline*}
An inductive hypothesis that $R(k-1,\ell-1)\ge 0$, Lemma \ref{lem:boundsabb}, and the positivity of $r(x_1y_1 \mid k-1,\ell-1)$
  then show
$$R(k,l) > \frac 1{\binom k2}\left (  1 + 0   +\binom{k-2}2 0 -(k-2)\frac 1{k-2}\right )  =  0. $$
\end{proof}

\begin{lemma} \label{lem:abb} Consider a species tree with topology $((T,a),(b_1,b_2))$,
where $T$ is a subtree on at least three taxa, one of which is $c$. Suppose the edge above $(T,a)$
has positive length. Then under the multispecies coalescent model,
$$\PP(Sp(ac))+\PP(Sp(ab_2c))-\PP(Sp(b_1c))-\PP(Sp(b_1b_2c))>0.$$
\end{lemma}
\begin{proof}
Let $\rho$ denote the root of the species tree, and $v$ the MRCA of the taxa on $T$ and $a$.

To establish the claimed inequality, it is enough to show it holds when conditioned on 
whether the $b_1$ and $b_2$ lineages have coalesced before reaching $\rho$ or not.  If they have 
coalesced below $\rho$ to form a single lineage, then the events $Sp(ab_2c)$ and 
$Sp(b_1c)$ have probability zero. Thus using $b$ for $b_1b_2$ we wish to show
$$\PP(Sp(ac))-\PP(Sp(bc))>0.$$
This follows immediately from Proposition \ref{prop:acbc}.

We henceforth condition on the event that the lineages from $b_1$ and
$b_2$ are distinct at $\rho$.  Noticing that all four probabilities in
the expression of interest are 0 if the $c$ lineage coalesces with any lineage
below $v$, we further condition on the event that the $c$ lineage is distinct at $v$, so there are $k\ge 3$ distinct lineages at $v$, and
that $\ell$ coalescent events occur on the edge above $v$.  Calling
this event $\mathcal C = \mathcal C(k, \ell)$,
\begin{align*}
\PP(Sp(ac)\mid \mathcal C)&=r(x_1x_2\mid k,\ell),\\
\PP(Sp(ab_2c)\mid \mathcal C)&=r(x_1x_2y_1\mid k, \ell),\\
\PP(Sp(b_1c)\mid \mathcal C)&=r(x_1y_1\mid k,\ell),\\
\PP(Sp(b_1b_2c)\mid \mathcal C)&=r(x_1y_1y_2\mid k,\ell).
\end{align*}
From Lemma \ref{lem:ineqabb} we find that conditioned on $\mathcal C$
the expression of interest is strictly positive for $k\ge 4$, $k >
\ell \ge 1$, and zero for $k=3$, $\ell=0,1,2$ and $k\ge 4$,
$\ell=0$. Weighting the conditioned expressions by the probabilities
of the $\mathcal C$ and summing we get the unconditioned expression.
Since $T$ has at least 3 taxa and the branch length above $v$ has
positive length, some of the summands corresponds to the event
$\mathcal C(k,\ell)$ with $k\ge 4$, $k>\ell\ge 1$; thus the full
expression is positive.
\end{proof}

Finally, Proposition \ref{prop:abbc} follows from Theorem
\ref{thm:cladesplitinv}, Lemma \ref{lem:bba} and Lemma \ref{lem:abb}.

\subsection{Proof of Proposition \ref{prop:aabbc}} \

To establish Proposition \ref{prop:aabbc}, we first extend the results
of Lemma \ref{lem:recurabb}, and those that follow it, to splits of
size 4.

A proof of the following is left to the reader.

\begin{lemma} \label{lem:recuraabb}\ 
\begin{enumerate}
\item $r(x_1x_2x_3y_1\mid 3, 0)=1$,
\item $r(x_1x_2x_3y_1\mid k, 0)=\frac 6{\binom {k+2}2} r(x_1x_2y_1\mid k-1, 0)+\frac {\binom{k-2}2}{\binom {k+2}2}r(x_1x_2x_3y_1\mid k-1,0)$ for $k\ge 4$,
\item $r(x_1x_2x_3y_1\mid 4,\ell)=\frac 1{2}r(x_1x_2y_1\mid 3,\ell-1)$ for $\ell = 1, 2, 3$,
\item $r(x_1x_2x_3y_1\mid k,\ell)=\frac 3{\binom k2}r(x_1x_2y_1\mid k-1,\ell-1) +\frac {\binom{k-3}2}{\binom k2}r(x_1x_2x_3y_1\mid k-1,\ell-1)$ for $k\ge 5$, $k>\ell\ge 1$,
\item $r(x_1x_2y_1y_2\mid 3, \ell)=1$ for $\ell=0,1,2$,
\item $r(x_1x_2y_1y_2\mid k, \ell)=\frac 1{\binom k2}r(x_1y_1y_2\mid k-1,\ell-1) +\frac {\binom{k-2}2}{\binom k2}r(x_1x_2y_1y_2\mid k-1,\ell-1)$ for $k\ge 4$, $k>\ell\ge 1$,
\end{enumerate}
\end{lemma}

\begin{lemma}\label{lem:boundsaabb}
Let $U(k,\ell)=$
$$r(x_1y_1\mid k, \ell)+r(x_1x_2y_1\mid k, \ell)
+r(x_1y_1y_2\mid k, \ell)+r(x_1x_2y_1y_2\mid k,\ell).
$$
Then
$U(k,\ell)< \frac 1{k-2}$ for $k\ge 4$, $k>\ell \ge 0$.
\end{lemma}

\begin{proof}
We first take up the case that $\ell = 0$, and 
observe by Lemmas \ref{lem:recurabb} and  \ref{lem:recuraabb} that for $k\ge 4$,
\begin{align*}
U(k,0)&=\frac 1{\binom{k+2}2} \Big ( 1+(2k+3) \left (r(x_1y_1\mid k-1,0)+r(x_1y_1y_2 \mid k-1,0)\right)\\ 
&\qquad - r(x_1y_1y_2 \mid k-1,0)+\binom {k-2}2 U(k-1,0) \Big ).
\end{align*}
Since 
$$U(3,0)=\frac 15 +\frac15+\frac 15+1=\frac 85,$$ 
we see
$$U(4,0)=\frac1{15}\left(1+11\left( \frac 15 +\frac 15\right)-\frac 15+1\cdot \frac 85 \right )=\frac{34}{75}<\frac 1{4-2},$$
establishing the $k=4, \ell = 0$ case.
Proceeding inductively for $k\ge 5$, and using Lemma \ref{lem:boundsabb} \eqref{caseb:abbxyykESA}, we have
\begin{align*}
U(k,0)&<\frac 1{\binom{k+2}2} \left ( 1+(2k+3)\frac 1{k-2}  -0+\binom {k-2}2 \frac 1{k-3} \right )\\
&= \frac{k^2 + 2k + 6}{k^2+3k+2} \cdot \frac{1}{k-2}<\frac 1{k-2}.
\end{align*}

For $\ell > 0$, 
if $k\ge 4$, $k>\ell\ge 1$,  Lemmas \ref{lem:recurabb} and  \ref{lem:recuraabb}
show
\begin{multline}
U(k,\ell)=\frac 1{\binom k2}\bigg ( \binom{k-2}2 U(k-1,\ell-1)\\
+(k-1) \left(r(x_1y_1\mid k-1,\ell-1)+r(x_1y_1y_2\mid k-1,\ell-1)\right)  \Big).\label{eq:aabbtRkl}
\end{multline}
In particular, since
\begin{align*} 
 U(3,1)&=\frac 19+\frac 19 +\frac 13+1=\frac {14}9,\\
U(3,2)&=0+ 0+\frac 13+1=\frac 43,
\end{align*}
then
\begin{align*}
U(4,1)&=\frac 16\left(\frac 85+3\left (\frac 15+\frac 15\right )\right )=\frac {7}{15}<\frac 1{4-2},\\
U(4,2)&=\frac 16\left(\frac {14}9+3\left (\frac 19+\frac 13\right )\right )=\frac {13}{27}<\frac 1{4-2},\\
U(4,3)&=\frac 16\left(\frac 43+3\left (0+\frac 13\right )\right )=\frac 7{18}<\frac 1{4-2},
\end{align*}
providing, along with the cases with $\ell = 0$, the base cases for induction.
Now for $k\ge 5$, $k>\ell\ge 1$,
we see from equation \eqref{eq:aabbtRkl}, Lemma \ref{lem:boundsabb} \eqref{caseb:abbxyykESA},
and an inductive hypothesis that
\begin{align*}
U(k,\ell) &<\frac 1{\binom k2}\Big ( \binom{k-2}2\frac 1{k-3}+(k-1)\frac 1{k-2} \Big ) \\
&=\frac{k^2-2k+2}{k(k-1)(k-2)}<\frac 1{k-2}.
\end{align*}
\end{proof}

\begin{lemma}\label{lem:ineqaabb}
Let 
\begin{multline*}
S(k,\ell)=r(x_1x_2\mid k,\ell)+r(x_1x_2x_3\mid k,\ell)+r(x_1x_2x_3y_1\mid k, \ell)\\
-r(x_1y_1\mid k,\ell)-r(x_1y_1y_2\mid k,\ell)-r(x_1x_2y_1y_2\mid k,\ell).
\end{multline*}
Then  for $k=4$, $\ell=0, 1,2,3$ and for $k\ge 5$, $\ell=0$, $S(k,\ell)= 0$.
For $k\ge 5$, $k>\ell\ge 1$, $S(k,\ell)>0$.
\end{lemma}
\begin{proof}
Since for $k=4$, the events $Sp(x_1x_2)=Sp(x_3x_4y_1y_2)$, $Sp(x_1x_2x_3)=Sp(x_4y_1y_2)$, and 
$Sp(x_1x_2x_3y_1)=Sp(x_4y_2)$, so using exchangeability of the $x_i$ and of the $y_i$ lineages
we have
\begin{align*}
r(x_1x_2\mid 4,\ell)&=r(x_1x_2y_1y_2\mid 4,\ell),\\
r(x_1x_2x_3\mid 4,\ell)&=r(x_1y_1y_2\mid 4,\ell),\\
r(x_1x_2x_3y_1\mid 4,\ell)&=r(x_1y_1\mid 4,\ell),
\end{align*}
so
$S(4,\ell)=0$ for $\ell=0,1,2,3$.  
For $k\ge 5$, Lemma \ref{lem:recurabb} \eqref{case:abb0}  implies $S(k,0)=0$.

For $k\ge 5$, $k>\ell\ge 1$, using Lemmas \ref{lem:recurabb} and \ref{lem:recuraabb} we find
\begin{align*}S(k,\ell)&=
\frac 1{\binom k2}\Big ( \binom{k-3}2 S(k-1,\ell-1) -(k-3)U(k-1,\ell-1)+kR(k-1,\ell-1)\\
& \qquad  + 1+  2r(x_1y_1\mid k-1,\ell-1)+r(x_1y_1y_2\mid k-1,\ell-1)) \Big ).
\end{align*}
Using an inductive hypothesis, Lemmas \ref{lem:boundsaabb}, and \ref{lem:ineqabb} 
and non-negativity of probabilities, this implies
$$S(k,\ell)> \frac 1{\binom k2} \left (\binom{k-3}2\cdot 0-(k-3)\frac 1{k-3}+k\cdot 0+1+ 2 \cdot 0+0\right )=0.$$
\end{proof}

\begin{proof}[Proof of Proposition \ref{prop:aabbc}]
  On the species tree $((T,(a_1,a_2)),(b_1,b_2))$ let $\rho$ denote
  the root, $v$ the MRCA of the taxa on $T$ and the $a_i$, and let $c$
  be a taxon on $T$.  We first show that since the edge above $v$ has
  positive length, then
\begin{multline}
\PP(Sp(a_1c))+\PP(Sp(a_1a_2c))+\PP(Sp(a_1b_2c))+\PP(Sp(a_1a_2b_2c))\label{eq:aabbexp}\\
-\PP(Sp(b_1c))- \PP(Sp(b_1a_2 c))- \PP(Sp(b_1b_2c))- \PP(Sp( b_1a_2 b_2c))>0.
\end{multline}

 To establish this, it is enough to show it holds when conditioned on whether  or not the 
 $a_1$ and $a_2$ lineages have coalesced before reaching $v$, and whether  or not the
 $b_1$ and $b_2$ lineages have coalesced before reaching 
$\rho$. If both pairs have coalesced 
 in this way, then the events $Sp(a_1c)$, $Sp(a_1b_2c)$, $Sp(a_1a_2b_2c)$, $Sp(b_1c)$, 
 $Sp(b_1a_2 c)$, and $Sp(b_1 a_2  b_2c)$ all have probability zero. Using 
 $a$ for $a_1a_2$ and $b$ for $b_1b_2$ we need only show
$$\PP(Sp(ac))-\PP(Sp(bc))>0.$$
This follows immediately from Proposition \ref{prop:acbc}. Similarly, the cases in 
which exactly one of  the pairs of $a_1, a_2$ lineages or $b_1, b_2$ lineages have coalesced in the 
population immediately ancestral to their respective MRCAs follow from Proposition \ref{prop:abbc}.

We henceforth condition on the event that the $a_i$ lineages are distinct at $v$ and 
the $b_i$ lineages are distinct at $\rho$.
Noticing that all eight probabilities in the expression of interest are 0 if the $c$ lineage coalesces with any 
lineage below $v$, we further condition on the $c$ lineage being distinct at $v$ (so
there are $k\ge 4$ lineages in total entering the population above $v$) and $\ell$ coalescent 
events occur between $v$ and $\rho$. 

Then, with $\mathcal C = \mathcal C (k,\ell)$ denoting the event that these conditioning requirements are met, 
\begin{align*}
\PP(Sp(a_1c)\mid \mathcal C)&=r(x_1x_2\mid k,\ell),\\
\PP(Sp(a_1a_2c)\mid \mathcal C)&=r(x_1x_2x_3\mid k,\ell),\\
\PP(Sp(a_1b_2c)\mid \mathcal C)&=r(x_1x_2y_1\mid k,\ell),\\
\PP(Sp(a_1a_2b_2c)\mid \mathcal C)&=r(x_1x_2x_3y_1\mid k,\ell),\\
\PP(Sp(b_1c)\mid \mathcal C)&=r(x_1y_1\mid k, \ell),\\
\PP(Sp(b_1a_2  c)\mid \mathcal C)&=r(x_1x_2y_1\mid k,\ell),\\
\PP(Sp(b_1b_2c)\mid \mathcal C)&=r(x_1y_1y_2\mid k,\ell),\\
\PP(Sp(b_1a_2 b_2c)\mid \mathcal C)&=r(x_1x_2y_1y_2\mid k,\ell).\\
\end{align*}
After substituting these in to the expression in \eqref{eq:aabbexp},
from Lemma \ref{lem:ineqaabb} we see that when conditioned on $\mathcal C$ it is strictly 
positive for $k\ge 5$, $k\ge \ell\ge 1$ and zero for $k=4$, $\ell=0,1,2,3$ and for $k\ge 5$, $\ell=0$. 
Thus weighting the conditioned expressions by the 
probabilities of the $\mathcal C$ and summing
over all relevant $k$ and $\ell$, we see the unconditioned inequality \eqref{eq:aabbexp} holds 
since  $T$ has at least 3 taxa so  summands with $k\ge 5$, $\ell\ge 1$ are present. 

Interchanging the $a_i$ and $b_i$ in inequality \eqref{eq:aabbexp} shows the negativity 
of the expression on the tree $((T,(b_1,b_2)),(a_1,a_2))$. Since its vanishing on the tree 
$(T,((a_1,a_2),(b_1,b_2)))$ 
was shown in Theorem \ref{thm:cladesplitinv}, the proof is complete.
\end{proof}

\bibliographystyle{alpha}
\bibliography{Split_probs}

\end{document}